\documentclass[a4paper,11pt]{article}

\pdfoutput=1

\usepackage[margin=1in]{geometry}

\usepackage{titling}
\usepackage[affil-it]{authblk}

\usepackage[english]{babel}
\usepackage{setspace}
\usepackage{amsmath}

\DeclareMathOperator*{\argmin}{arg\,min}
\usepackage{relsize}
\usepackage{hyperref}
\DeclareFontFamily{U}{mathx}{\hyphenchar\font45}
\DeclareFontShape{U}{mathx}{m}{n}{<-> mathx10}{}
\DeclareSymbolFont{mathx}{U}{mathx}{m}{n}
\DeclareMathAccent{\widebar}{0}{mathx}{"73}

\usepackage{array}
\usepackage{dsfont}
\usepackage{graphicx}
\usepackage{caption}
\usepackage{subcaption}
\usepackage{graphics}

\usepackage[round]{natbib}
\usepackage{amssymb}
\usepackage{mathtools}
\usepackage{booktabs,caption}
\usepackage{tabularx}
\newcolumntype{Y}{>{\raggedleft\arraybackslash}X}
\newcolumntype{Z}{>{\centering\arraybackslash}X}

\usepackage{pgfplots}
\pgfplotsset{compat=1.10}
\usepackage{pdflscape}
\extrafloats{100}
\usepackage{float}
\usepackage{MnSymbol} 

\usepackage{epstopdf}
\usepackage[tablesfirst,nolists]{endfloat}
\usepackage{threeparttable}
\usepackage{rotating}
\DeclareDelayedFloatFlavor{sidewaystable}{table}
\usepackage{tabularx}
\usepackage{amsthm}

  \usepackage{natbib}
\newtheorem{mylem}{Lemma}

\theoremstyle{definition}{
\newtheorem{mydef}{Definition}}

\theoremstyle{definition}{
}

\theoremstyle{definition}{
}

\newtheorem{innercustomthm}{Condition}
\newenvironment{customcond}[1]
  {\renewcommand\theinnercustomthm{#1}\innercustomthm}
  {\endinnercustomthm}

\usepackage{chngcntr}
\usepackage{apptools}
\AtAppendix{\counterwithin{mythm}{section}}
\AtAppendix{\counterwithin{myprop}{section}}
\AtAppendix{\counterwithin{mylem}{section}}
\AtAppendix{\counterwithin{mydef}{section}}
\AtAppendix{\counterwithin{mycor}{section}}

\begin{document}
\title{Expected Shortfall LASSO}

\author{
Sander Barendse}
\affil{University of Amsterdam}

\date{This version: January 10, 2024 \\ First version: July 3, 2023}
\maketitle

\begin{abstract}
\footnotesize
\noindent 
We propose an $\ell_1$-penalized estimator for high-dimensional models of Expected Shortfall (ES). The estimator is obtained as the solution to a least-squares problem for an auxiliary dependent variable, which is defined as a transformation of the dependent variable and a pre-estimated tail quantile. Leveraging a sparsity condition, we derive a nonasymptotic bound on the prediction and estimator errors of the ES estimator, accounting for the estimation error in the dependent variable, and provide conditions under which the estimator is consistent. Our estimator is applicable to heavy-tailed time-series data and we find that the amount of parameters in the model may grow with the sample size at a rate that depends on the dependence and heavy-tailedness in the data. In an empirical application, we consider the systemic risk measure CoES and consider a set of regressors that consists of nonlinear transformations of a set of state variables. We find that the nonlinear model outperforms an unpenalized and untransformed benchmark considerably.

\end{abstract}

\newpage

\section{Introduction}

Expected shortfall (ES) has quickly become the standard regulatory measure for market risk, as the Basel III framework of the Basel Committee on Banking Supervision has required banks to base their internal risk models on it since 2016 \citep{basel2016}. Defined as the expectation of the returns that exceed Value-at-Risk (VaR), another risk measure which corresponds to a tail quantile, the global regulatory body considers ES to provide a more prudent capture of tail risk compared to VaR, partly due to theoretical considerations outlined in \citet{artzner1999coherent} among others. Accurate forecasting of ES is therefore of great importance to financial institutions.

Models of ES are often parsimonious in nature due to the low signal-to-noise ratio in financial returns, which complicates estimation of large models. The call for parsimony is often amplified by the low frequency of return observations close to or in exceedance of the VaR, such that the researcher has few observations in their data set that are highly informative in the estimation of ES models. And yet, we know that the risk in returns is connected to  various fundamental variables (see Chapter 4 in \citet{andersen2013financial} and references therein). Moreover, dependence on these fundamental variables may be nonlinear, which, in turn, increase the amount of parameters in the model. For instance, the leverage effect states there is nonlinear relationship between stock market volatility and future returns, see, e.g. \citet{CAMPBELL1992}. The development of estimation methods that can handle large ES models is therefore important.

In this paper we develop a semiparametric estimator of linear ES models with many regressors, with the express purpose of estimating models that may rely on many (transformations of) fundamental variables. The estimator is an $\ell_1$-penalized least squares estimator (LASSO) that utilizes pre-estimated VaR predictions and leverages a sparsity condition in the model. With the VaR predictions, an auxiliary dependent variable is created that may be regressed on the set of explanatory variables to obtain an ES estimator, as shown by \citet{barendse2020}. We derive a nonasymptotic bound on the prediction and estimation errors, and provide conditions under which consistency is obtained. These conditions allow the amount of regressors $p$ to grow with and potentially be much larger than the sample size $T$. Although theoretical properties of LASSO estimators for time-series data have been developed in the literature, our estimator and results are novel because we explicitly account for the prediction error in the conditional variable that appears due to its dependence on the pre-estimated quantile predictions.

We use a penalty function that accounts for differences in the scale of the regressors which has not been studied in the literature before for least squares estimators, to our knowledge, and which is borrowed from the penalized quantile regression literature, see \citet{belloni2011}. Although accounting for scale differences is an appealing property on its own, we also prefer this penalty function as it means the penalization function of the coefficients in the VaR and ES estimation steps is equivalent. Indeed, we rely on the quantile regression estimator of \citet{belloni2023high} (from hereon BCMPW) to obtain the VaR predictions in the first step, although our results apply to any quantile estimator that satisfies the conditions.

As financial data is often dependent over time and heavy-tailed, we derive our results under conditions that allow for these properties. Specifically, we allow the data to be $\beta$-mixing sequences with finite moments of a certain order. To obtain our results, we develop a Fuk-Nagaev inequality for $\beta$-mixing sequences, which may be of independent interest. By using the Fuk-Nagaev inequality in conjuction with the same blocking strategy as utilized in BCMPW to obtain approximately independent blocks of data, we obtain rates on the model parameters that are directly comparable to those for the penalized quantile regression estimator in BCMPW. This result is closely related to the Fuk-Nagaev inequality for $\tau$-mixing sequences in \citet{babii2022machine}, which relies on an alternative blocking strategy, and the tail inequality for sequences that satisfy a sub-Weibull property in \citet{Wong2020}.


We apply our method in a systemic risk analysis. In this analysis we generate predictions of the Conditional Expected Shortfall (CoES) to measure the risk spillover from the financial sector to the entire stock market. The CoES measure was introduced in \citet{adriancovar} as an extension of CoVaR, an alternative systemic risk measure which is commonly used but is less prudent than CoES from a similar argument to the one that favors ES over VaR. The estimation of CoES relies on three stages of quantile regression and one stage of ES regression. Like \citet{adriancovar} we condition on seven lagged fundamental variables, but we extend their set of fundamental variables by also considering nonlinear transformations. Specifically, we use each the Chebyshev polynomials to transform each of the fundamental variables, where the degree of polynomials used may increase with the sample size. Our results show that the penalized VaR and ES estimators outperform the unpenalized benchmark estimator that strictly uses the untransformed set of fundamental variables by a considerable margin in terms of out-of-sample prediction error. We observe that a moderate degree of Chebyshev polynomials is optimal in our sample and generates CoES measurements that are  more conservative than the benchmark and more responsive to news. Finally, we observe that the inclusion of nonlinear transformations of the fundamental variables results in a kind of leverage effect: risk predictions are less impacted by positive returns. 

In related literature, there are several papers that study the properties of LASSO estimators for linear mean models in a  time-series setting. These papers consider non-estimated conditional variables, and therefore do not apply to the ES estimation method proposed in this paper. Closely related are \citet{babii2022machine}, who study the properties of the group-LASSO estimator for $\tau$-mixing data, and \citet{Wong2020}, who consider the LASSO estimator for $\beta$-mixing sequences that satisfy a sub-Weibull property, which is stronger than the moment conditions we impose. \citet{medeiros20161} consider the adaptive LASSO estimator for models with martingale difference sequence errors and \citet{adamek2023lasso} extend their results to the near-epoch dependence error case. \citet{masini2022regularized} study the case with mixingales.
\citet{wu2016}, \citet{uematsu2019high}, and \citet{chernozhukov2021lasso} consider LASSO estimators for Bernoulli shift data. Finally, 
\citet{nardi2011autoregressive}, \citet{kock2015oracle}, and \citet{basu2015} develop LASSO estimators under more restrictive conditions on the data. \citet{hsu2008subset} and \citet{wang2007regression} develop LASSO estimators for scenarios in which the number of variables is smaller than the sample size.
In contemporaneous research, \citet{xuming2023} also develop an estimator of high-dimensional ES models based on the least-squares procedure in \citet{barendse2020}. Instead of a penalized estimator, they consider a Huberized robust estimator and develop their results under independence and sub-Gaussianity of the data.

The rest of the paper is organized as follows. Section \ref{sec:methods} develops the estimator and derives a nonasymptotic bound and and consistency results. Section \ref{sec:emp:results} contains the empirical study on CoES estimation. Section \ref{sec:simulation} includes a Monte Carlo simulation study. Section \ref{sec:fuknagaev} develops the Fuk-Nagaev inequality. The Appendix includes proofs and additional results for the empirical analysis and Monte Carlo study.

\section{Methods}\label{sec:methods} 

\textbf{Notation:} 
Define, for $p\in \mathbb{N}$, $[p] = \{1,\ldots,p\}$. For $b \in \mathbb{R}^p$, we denote the $\ell_q$-norm as $\|b\|_q = (\sum_{i\in [p]} |b_i|^q)^{1/q}$, for $q \geq 1$, and $\|b\|_\infty = \max_{i\in[p]} |b_i|$ for $q=\infty$. For $a,b\in \mathbb{R}$, we denote $a\vee b = \max(a,b)$ and $a\wedge b = \min(a,b)$. For some vector $b \in \mathbb{R}^p$ and $V\subset [p]$ some index set, we denote by $b_V \in \mathbb{R}^p$ the vector for which ${b_V}_i = b_i$ if $i\in V$ and ${b_V}_i= 0 $ if $i \notin V$. Finally, for sequences $a_T$ and $b_T$, we write $a_T \lesssim b_T$ if there exists a constant $C > 0$ such that $a_T \leq C b_T$, for all $T\geq 1$, and $a_T \asymp b_T$ if $a_T \lesssim b_T$ and $b_T \lesssim a_T$.

\subsection{The model}

We consider the following data generating process for the conditional variable $Y_t$, given the regressor vector $X_t$:
\begin{align}\label{eq:linearmodel}
Y_t = X_t'\alpha^0(U_t),
\end{align}
where $\{U_t\}$ is a sequence of random errors satisfying $U_t | X_t \sim \text{Unif}(0,1)$, $\{X_t\}$ is a sequence of random vectors satisfying $X_t \in \mathcal{X}\subseteq\mathbb{R}^p$, and $\alpha^0(\cdot)$ is some measurable $p$-dimensional vector functional on $(0,1)$. To obtain increasing quantiles we impose, for any $x \in \mathcal{X}$, the function $x'\alpha^0(u)$ is strictly increasing in $u\in(0,1)$. In forecasting scenarios, $X_t$ contains lagged variables. Finally, we note that this model nests the location-scale model, see Section \ref{sec:simulation}.

In this model the $\tau$-quantile of $Y_t$ conditional on $X_t$, denoted $Q_t(\tau)$, has the functional form
\[Q_t(\tau) = X_t' \alpha^0(\tau),\]
for some quantile level $\tau\in(0,1)$. Moreover, the conditional ES, denoted $ES_t(\tau)$, has the functional form
\[ES_t(\tau) := E\left[Y_t \ | \ Y_t \leq Q_t(\tau), X_t \right] =  X_t' \gamma^0(\tau),\]
where $\gamma^0(\tau) = \int_0^\tau \alpha^0(u) du$ denotes the ES coefficient vector.

From the above specifications, we note that the DGP in (\ref{eq:linearmodel}) allows the regressors $X_t$ to influence the conditional variables $Y_t$ differently depending on the level of $Y_t$, since $\alpha^0(\tau)$ and $\gamma^0(\tau)$ are dependent on $\tau$. Moreover, the regressor vector $X_t$ may include nonlinear transformations of some subset of regressors $Z_t$. Hence, the DGP allows the regressors $Z_t$ to influence $Y_t$ differently depending on the magnitude of (the elements of) regressor vector $Z_t$. Indeed, conditional on $U_t = \tau$, $\frac{\partial Y_t}{\partial Z_t} = \frac{\partial X_t'}{\partial Z_t} \alpha^0(\tau)$, which may depend on $Z_t$ if $X_t$ contains nonlinear transformations of $Z_t$. As an example, the credit spread may influence the market return differently depending on the level of the credit spread and the level of the market return.

From now on we fix $\tau$ and remove reference to it for notational convenience. Specifically, we use $\alpha^0 = \alpha^0(\tau)$,$\gamma^0 = \gamma^0(\tau)$, $Q_t = Q_t(\tau)$, and $ES_t(\tau) = ES_t$. 

\subsection{The estimator}
Before we introduce the estimator of the ES coefficient vector $\gamma^0$, we introduce some additional notation to describe the data set.
For a sample of $T$ observations, let $Y = (Y_1,\ldots,Y_T)'$ denote the vector of conditional variables and $X = [X_1,\ldots,X_p]$ the regressor matrix, where $X_i = (X_{i1},\ldots,X_{iT})'$. Moreover, define the auxiliary conditional variable $\tilde{Y} = (\tilde{Y}_1,\ldots,\tilde{Y}_T)'$, where $\tilde{Y}_t = Q_t + \frac{1}{\tau} \mathds{1}(Y_t < Q_t) (Y_t - Q_t)$. A feasible counterpart to $\tilde{Y}$ is given by $\hat{Y} = (\hat{Y}_1,\ldots,\hat{Y}_T)'$ , where $\hat{Y}_t = \hat{Q}_t + \frac{1}{\tau} \mathds{1}(Y_t < \hat{Q}_t) (Y_t - \hat{Q}_t)$, with $\hat{Q}_t = X_t'\hat{\alpha}$ and $\hat{\alpha}$ some estimator of $\alpha^0$.

It can be shown that the ES parameters $\gamma^0$ optimize a (population) least squares problem with conditional variable $\tilde{Y}_t$ and regressor vector $X_t$. This follows from the regression errors $\varepsilon = (\varepsilon_1,\ldots,\varepsilon_T)' := \tilde{Y} - X'\gamma^0$ satisfying the condition $E[\varepsilon_t|X_t] = 0$ by definition of the expected shortfall above. See \citet{barendse2020} for details.

The above argument naturally leads to an estimator of $\gamma^0$ that solves a sample counterpart to the least squares problem using the feasible auxiliary conditional variables $\hat{Y}$:
\begin{align}\label{eq:ES-estimator}
  \hat{\gamma} := \argmin_{\gamma} \frac{1}{T}\|\hat{Y} - X\gamma\|_2^2 + \lambda \|\gamma\|_{1,T},
\end{align}
with $\|\gamma\|_{1,T} := \sum_{i=1}^p \hat{\sigma}_i |\gamma_i|$, $\hat{\sigma}_i^2 := \frac{1}{T}\sum_{t=1}^T X_{it}^2$. 
The second term in the optimization problem in (\ref{eq:ES-estimator}) is a penalty term that depends on the magnitude of the elements of the coefficient vector $\gamma$ and the penalty level $\lambda \geq 0$. 
The norm $\|\cdot\|_{1,T}$ is introduced in \citet{belloni2011} and used in their quantile regression problem with LASSO penalization for iid data. \citet{belloni2023high} also use it in their LASSO quantile regression for time-series data. It effectively standardizes the data by giving the coefficients corresponding to regressors with high dispersion a larger penalty than coefficients of regressors with low dispersion. We utilize the equivalent norm, such that the LASSO estimator of the quantile coefficients $\alpha^0$ uses the same penalization function as the LASSO estimator of the ES coefficients $\gamma^0$.


\subsection{Nonasymptotic bound}
We derive a nonasymptotic bound on the prediction error $\frac{1}{T}\|X(\hat{\gamma}-\gamma^0) \|_2^2$ and the estimation error $\|\hat{\gamma}-\gamma^0\|_1$ under the following assumptions on the data.

\begin{customcond}{A}\label{cond:Variances}
It holds $\sigma_i^2 := E[X_j^2] = 1$, for all $i\in[p]$. Moreover, $\Sigma := E[X'X]$ has full rank.
\end{customcond}

The first part of Condition \ref{cond:Variances} is a simplifying assumption. In alternative cases, the prediction error $\|\hat{Y} - X\gamma\|_2^2$ does not change. The estimation error does change, but is bounded by the rate for the standardized case multiplied by $\widebar{\sigma^{-1}} := \max_{i\in[p]} \sigma_i^{-1}$, see Appendix \ref{appx:conditionvariances}.

Let $S_0 \subset \{1,\ldots,p\}$ denotes the sparsity set of $\gamma^0$, such that $\gamma_i = 0$ if $i\notin S_0$. The cardinality of $S_0$ is denoted by $s_0$.  Also define the cone $A = \{\gamma : \|\gamma_{S_0^c}\|_1 \leq C_0 \|\gamma_{S_0}\|_1\}$, for some constant $C_0 \geq 5/3$. 

\begin{customcond}{REC}\label{cond:REC}
There exists a constant $\phi_0^2 >0 $, such that for each $\gamma \in A$, $\|\gamma_{S_0}\|_1^2 \leq \frac{s_0}{\phi_0^2} \gamma' \Sigma \gamma$.
\end{customcond}

Condition \ref{cond:REC} is a restricted eigenvalue condition using the population covariance matrix $\Sigma$, see \citet{bickel2009simultaneous}. This condition is standard in the literature.

Lemma \ref{lem:nonasymp-bound} below provides nonasymptotic bounds on the estimation error $\|\hat{\gamma}-\gamma^0\|_1$ and the prediction error $\frac{1}{T}\|X(\hat{\gamma}-\gamma^0) \|_2^2$ that depend on the sparsity of the ES model and the penalty level, per usual, as well as on the prediction error contained in the feasible auxiliary conditional variables $\hat{Y}$. The bound holds on the intersection of the following events: $\mathcal{S} := \{\max_{i\in[p]}\left|\hat{\sigma}_i - \sigma_i\right| \leq \frac{1}{4}\}$; $\mathcal{T} := \{ \max_{i\in[p]} \frac{2}{T}\left|\varepsilon'X_j \right|  \leq \lambda_0\}$; and $\mathcal{U} := \{ \max_{i,j\in [p]} \left| \frac{1}{T}\sum_{t=1}^TX_{it} X_{jt} - E[X_{it} X_{jt}] \right|  \leq \lambda_1\}$, for positive constants $\lambda_0$ and $\lambda_1$.


\begin{mylem}\label{lem:nonasymp-bound}
  Let Conditions \ref{cond:Variances} and \ref{cond:REC} be satisfied. 
  On $\mathcal{T}  \cap \mathcal{S} \cap \mathcal{U}$, with $\lambda_0 \leq  \frac{6C_0 - 10}{15(1+C_0)} \lambda$ and $\lambda_1 \leq \frac{\phi_0^2}{2s_0 (1+C_0)^2}$, there exists some universal constant $C>0$ such that
\[\frac{1}{T}\|X(\hat{\gamma}-\gamma^0) \|_2^2  + \lambda \|\hat{\gamma}-\gamma^0\|_1 
\leq 
C \left( \lambda^2 s_0 \vee \frac{1}{T} \|\hat{Y} - \tilde{Y}\|_2^2\right).\]
\end{mylem}

\begin{proof}
See Appendix \ref{appx:conditionvariances}.
\end{proof}

\subsection{Asymptotics}\label{sec:asymptotics}

In this section we provide conditions under which the prediction error converges to zero in probability and the estimator is consistent, i.e.~$\|X(\hat{\gamma}-\gamma^0) \|_2^2/T = o_P(1)$ and $\|\hat{\gamma}-\gamma^0\|_1 = o_P(1)$. The result is obtained by showing that the event $\mathcal{T}  \cap \mathcal{S} \cap \mathcal{U}$ appearing in Lemma \ref{lem:nonasymp-bound} occurs with high probability and that the bound in Lemma \ref{lem:nonasymp-bound} converges to zero under the imposed conditions.

\begin{customcond}{D}\label{cond:mixing} For all $i,j\in[p]$, the sequences $\{X_{it} \varepsilon_t  \}$ and $\{X_{it} X_{jt}\}$ are strictly stationary and $\beta$-mixing with coefficients bounded by $\bar\beta(t)$.
\end{customcond}

Condition \ref{cond:mixing} allows for time-series dependence in the data. A formal definition of $\beta$-mixing is provided in Section \ref{sec:fuknagaev}. In Condition \ref{cond:rates} below we impose a bound on $\bar\beta(t)$.

\begin{customcond}{MOM}\label{cond:moments}
For some $q \geq 2$ and all $i,j\in[p]$, $E|X_{i1}\varepsilon_1|^q <\infty$ and
$E|X_{i1}X_{j1}|^q <\infty$.
\end{customcond}

Condition \ref{cond:moments} is imposed to allow for $X_{i1}\varepsilon_1$ and $X_{i1}X_{j1}$, $i,j\in[p]$, to not be sub-Gaussian, sub-exponential random variables. In financial settings, the state variables and idiosyncratic errors may be heavy-tailed, such that the sub-Gaussianity or sub-exponentiality property is restrictive. 
Condition \ref{cond:moments} may be rewritten to allow for different order of moment bounds on the random variables $X_{i1}\varepsilon_1$ and $X_{i1}X_{j1}$, $i,j\in[p]$, at the cost of more involved notation. This condition is also weaker than imposing the sub-Weibull property of \citet{Wong2020}.

\begin{customcond}{R}\label{cond:rates}
Let $\bar\beta(k) \leq B k^{-\mu}$, for some constants $B\geq 0$ and $\mu > 2$. Moreover, let $a_T = \lceil T^{1/(1+\mu')}\rceil$ and $d_T = \lfloor T/(2a_T)\rfloor$ for some $\mu'\in(0,\mu)$. Finally, 
\begin{enumerate}
  \item $p d_T a_T^{-\mu} = o(1)$ if $B>0$;
  \item $s_0  \frac{(pa_T)^{2/q}}{d_T^{(q-1)/q}} = o(1)$; 
  \item $s_0  \frac{\sqrt{\log(pa_T)}}{\sqrt{d_T}} = o(1)$;
  \item $\|\hat{Y} - \tilde{Y}\|_2^2/T = o_P(1)$.
  \end{enumerate}
\end{customcond}

Finally, Condition \ref{cond:rates} imposes rates on the parameters $s_0$, $p$, $T$, and the prediction error $\|\hat{Y} - \tilde{Y}\|_2^2/T$.
Conditions \ref{cond:rates}.1 is a consequence of the time-series features of the data and need not be imposed for iid data.
Condition \ref{cond:rates}.2 reflects the polynomial tail part in the Fuk-Nagaev inequality for the maxima of high-dimensional sums and is required if $X_{i1}\varepsilon_1$ and $X_{i1}X_{j1}$ are not sub-Gaussian or sub-exponential random variables, for any $i,j\in[p]$. From Condition \ref{cond:moments} we observe that, if the value of $q$ increases, such that the random variables satisfy stricter moment conditions, Condition \ref{cond:rates}.2 becomes less restrictive on the rates of $s_0$ and $p$. Condition \ref{cond:rates}.3 is the standard exponential rate imposed in the LASSO problem if $X_{i1}\varepsilon_1$ and $X_{i1}X_{j1}$ are iid sub-Gaussian or sub-exponential random variables, for each $i,j\in[p]$. Finally, Condition \ref{cond:rates}.4 requires that the estimation error in the auxiliary conditional variable converges to zero in probability.

Condition \ref{cond:rates} introduces the terms $a_T$ and $d_T$, which depend on $T$ and a parameter $\mu'$. Moreover, the parameter $\mu$ describes the dependence in the data and is large if dependence is low. The value of $\mu'$ may be chosen close to $\mu$, such that $d_T$ may be set approximately equal to $T$ if dependence is low. Given the choice of penalty parameter $\lambda$ below, there is a trade-off between the convergence rates of the estimator and predictions (see Lemma \ref{lem:nonasymp-bound}) and the rates of the parameters. The convergence rate improves for large $d_T$ and small $a_T$, such that we like to choose $\mu'$ large. On the other hand, Condition \ref{cond:rates}.1 restricts $p$ more for small $d_T$ and large $a_T$, such that we like to choose $\mu'$ small. This trade-off is mediated by large $\mu$ (low independence), since it reduces the need for small $\mu'$ in Condition \ref{cond:rates}.1.

As an example, take the high dependence case $\mu = 3$. If we choose $\mu' = 1$, then Condition \ref{cond:rates}.1 requires $p = o(T^{\mu-\mu'}) = o(T^2)$. Moreover, Condition \ref{cond:rates}.2 requires $s_0^q p^2 T^{2-\mu'(q+1)/(1+\mu')} = s_0^q p^2 T^{2-(q+1)/2} = o(1)$. On the other hand, if we choose $\mu' = 2$, we require $p = o(T^{\mu-\mu'}) = o(T)$ and $s_0^q p^2 T^{2-2(q+1)/3} = o(1)$, such that the trade-off between parameter rates and convergence rates is clearly visible. From the example, it is also clear that less dependence ($\mu \gg 2$) or less heavy-tailedness ($q \gg 2$) in the data will improve parameter rates and convergence rates.
 
Lemma \ref{lem:asymp} below gives a consistency result on the estimation and prediction errors, for penalty levels that satisify
\[\lambda \asymp \frac{(pa_T)^{2/q}}{d_T^{(q-1)/q}}\vee \frac{\sqrt{\log(pa_T)}}{\sqrt{d_T}}.\]

\begin{mylem}\label{lem:asymp} Under the conditions of Lemma \ref{lem:nonasymp-bound} and Conditions \ref{cond:mixing}, \ref{cond:moments}, and \ref{cond:rates}, it follows
that 
\[\frac{1}{T}\|X(\hat{\gamma}-\gamma^0) \|_2^2 = o_P(1)\ \ \text{ and } \ \ \|\hat{\gamma}-\gamma^0\|_1 = o_P(1).\] 
\end{mylem}

\subsection{Quantile bound}

The results in Lemma \ref{lem:asymp} are established under the condition that the prediction error in the auxiliary variables converges to zero in probability, i.e.~$\|\hat{Y} - \tilde{Y}\|_2^2/T = o_P(1)$. In this section we first show that a nonasymptotic bound of $\|\hat{Y} - \tilde{Y}\|_2^2/T$ depends on the quantile level $\tau$ and the prediction error of the quantile predictions $\|\hat{Q} - Q\|_2^2/T$, where $\hat{Q} = (\hat{Q}_1,\ldots,\hat{Q}_T)'$ denotes a sample of quantile predictions and $Q = (Q_1,\ldots,Q_T)'$ denotes the sample of true quantiles. Second, we provide sufficient conditions under which $\|\hat{Y} - \tilde{Y}\|_2^2/T = o_P(1)$, if the quantile predictions $\hat{Q}$ are generated by the quantile regression estimator of BCMPW.

The upper-bound on $\|\hat{Y} - \tilde{Y}\|_2^2/T$ given in Lemma \ref{lem:augregress-quantile-bound} below increases as we consider values of the quantile level $\tau$ that are closer to zero. Such values are relevant in risk management, where $\tau$ is often set to $2.5\%$ or $5\%$. It also shows that the bound is dependent on the choice of quantile predictions $\hat{Q}$. Given that the prediction error $\|\hat{Q} - Q\|_2^2/T$ usually depends on the amount of regressors $p$, it follows that a trade-off exists between $p$ and $\tau$ in terms of the bounds on the estimation and prediction errors in Lemma \ref{lem:nonasymp-bound}.

\begin{mylem}\label{lem:augregress-quantile-bound}
$\|\hat{Y} - \tilde{Y}\|_2^2/T \leq \left(1+\frac{1}{\tau}\right)^2 \|\hat{Q} - Q\|_2^2/T.$
\end{mylem}
\begin{proof}
See Appendix \ref{proof:augregress-quantile-bound}.
\end{proof}

We consider quantile predictions that are generated using the penalized quantile regression estimator of BCMPW: 
\begin{align}
\hat{\alpha} := \argmin \sum_{t=1}^T \left(\tau - \mathds{1}(Y_t - X_t'\alpha)\right) (Y_t - X_t'\alpha) + \nu \|\alpha\|_{1,T},
\end{align}
with $\nu \asymp \sqrt{\log(pa_T)/d_T}$. The quantile predictions follow as $\hat{Q}_t = X_t'\hat\alpha$. 

Lemma \ref{lem:quant-consistency} below states that the prediction error $\|\hat{Q} - Q\|_2^2/T$ converges to zero in probability under Condition \ref{cond:BCMPW}. Condition \ref{cond:BCMPW} contains the relevant conditions under which $\hat{\alpha}$ is shown to be consistent in BCMPW. 
Lemmas \ref{lem:augregress-quantile-bound} and \ref{lem:quant-consistency} together imply $\|\hat{Y} - \tilde{Y}\|_2^2/T = o_P(1)$.

\begin{customcond}{QR}\label{cond:BCMPW}
The DGP satisfies the BCMPW conditions as stipulated in Appendix \ref{appx:belloni}.
\end{customcond}

\begin{mylem}\label{lem:quant-consistency} Under Condition \ref{cond:BCMPW}, it follows $\|\hat{Q} - Q\|_2^2/T = o_P(1)$.
\end{mylem}
\begin{proof}
See Appendix \ref{proof:quant-consistency}.
\end{proof}

\section{Empirical study: Systemic risk of the banking sector}

\subsection{Systemic risk measurement by CoES}\label{sec:coes}

We consider an application in systemic risk, in which we study the risk spillover from the banking sector to the entire market. To measure this relationship we make use of the CoES measure of \citet{adriancovar}. This measure was originally introduced to quantify risk spillover from one financial firm to the financial sector, but it is also suitable to measure the spillover from the banking sector to the entire market, as we do here. 

In \citet{adriancovar} CoES was introduced alongside the CoVaR measure. CoVaR has been widely used for systemic risk assessment, and yet (Co)ES is preferred over (Co)VaR from a theoretical perspective: VaR does not satisfy the coherence property for risk measures \citep{artzner1999coherent}, whereas ES does satisfy it. It is due to this argument and others that ES is now the preferred risk measure for financial institutions by the Basel Committee on Banking Supervision (see, e.g., \citet{basel2016}). 

We formally define CoES and related quantities. Given a universe of assets, consider the market return $R^{M}_t$ and the banking industry return $R^{I}_t$. The definition of CoES of the market relative to the banking industry relies on two Value-at-Risk measurements, $\text{VaR}^{I}_{t}$ and $\text{VaR}^{M}_{t}$, which are defined (implicitly) as the following tail quantiles:
\begin{align*}
\tau &= P\left(R^I_{t} \leq \text{VaR}^{I}_{t} \ \left | \right. \  Z_{t-1} \right)\,\,\, \text{ and } \,\,\,
\tau = P\left(R^M_{t} \leq \text{VaR}^{M}_{t} \ \left | \right. \ \left(R^I_t,Z_{t-1}'\right)' \right),
\end{align*}
with $Z_{t-1}$ denoting a vector of lagged state variables (and where we ignore the sign convention for Value-at-Risk). 

From the definitions it is clear that $\text{VaR}^{I}_{t}$ denotes the $\tau-$quantile of $R^{I}_t$ given the lagged state variables $Z_{t-1}$. The $\text{VaR}^{M}_{t}$ denotes the $\tau-$quantile of $R^M_t$ given the lagged state variables $Z_{t-1} $\textit{and} the \textit{contemporaneous} industry return $R^{I}_t$. The quantile level $\tau\in(0,1)$ should be taken small, e.g. $\tau = 2.5\%$.

The $\text{CoES}_{t}$ is defined as follows
\begin{align*}
\text{CoES}_{t} := E\left[R^M_{t} \ \left | \ \right. R^M_{t} \leq \text{VaR}^{M}_{t}, \ R^I_{t} = \text{VaR}^{I}_{t}, \ Z_{t-1} \right],
\end{align*}
and therefore denotes the ES of $R^M_{t}$ given the regressor vector $(R^I_{t}, Z_{t-1}')'$ and the event $\left\{R^I_t = \text{VaR}^{I}_{t}\right\}$. From the conditioning on the latter event, it becomes clear that the distress state in the definition of CoES is quantified by the banking industry return being at its Value-at-Risk conditional on the state variables $Z_{t-1}$.

Finally, $\Delta$CoES describes the marginal change in CoES moving from a normal state of the banking sector to a distress state, i.e. from event $\{R^{I}_t = \text{Median}^{I}_t\}$ to $\{R^{I}_t = \text{VaR}^{I}_{t}\}$, where the former event stipulates the banking return is at its median value given state variables $Z_{t-1}$. Formally,
\begin{align*}
\Delta\text{CoES}_{t} &= \text{CoES}_{t} - \text{CoES}^{\text{median}}_{t},
\end{align*}
where $\text{CoES}^{\text{median}}_{t} := E\left[R^M_{t} \ \left | \ \right. R^M_{t} \leq \text{VaR}^{M}_{t}, \ R^I_{t} = \text{Median}^{I}_{t}, \ Z_{t-1} \right]$. We recall that the median is defined as the quantile at quantile level $\tau=50\%$, such that its estimation proceeds similarly to that of $\text{VaR}^{I}_{t}$, but at $\tau = 50\%$ instead of 2.5\%.


\subsection{Model and estimation}
We model the market return $R^{M}_t$ and industry return $R^{I}_t$ as in $(\ref{eq:linearmodel})$:
\begin{align}\label{eq:empiricalmodel}
R^M_t &= \phi(R^{I}_t,Z_{t-1})'\alpha^M(U^M_t);\\
R^I_t &= \psi(Z_{t-1})'\alpha^I(U^I_t),
\end{align}
where the regressor vectors $\phi(\cdot)$ and $\psi(\cdot)$ are known vector functionals that (nonlinearly) transform their arguments. This model is highly flexible and allows for a nonlinear relationship between state variables, market return, and banking industry return. The relationship may also vary with quantile level $\tau\in(0,1)$. We elaborate on the choice of nonlinear transformations of the regressors in Section \ref{sec:dictionary}. Finally, we recall that (\ref{eq:linearmodel}) generalizes the location-scale model, which is commonly used in risk management. 

This model implies the following VaR and ES models:
\begin{align*}
\text{VaR}^{M}_{t} &=\phi(R^{I}_t,Z_{t-1})' \alpha^{M};\\
\text{VaR}^{I}_{t} &= \psi(Z_{t-1})' \alpha^{I};\\
\text{ES}_{t} &= \phi(R^{I}_t,Z_{t-1})' \gamma^{M},
\end{align*}
where we allow for implicit definitions of the coefficients for brevity. Finally, the CoES is given by
\begin{align*}
\text{CoES}_{t} &= \phi(\text{VaR}^{I}_{t},Z_{t-1})' \gamma^{M}.
\end{align*}

\subsection{Data, sample split, and estimation procedure}\label{sec:data-samplesplit-estimation}

We obtain \textit{weekly} banking industry return and market return for the US stock market from the website of Kenneth French (\url{https://mba.tuck.dartmouth.edu/pages/faculty/ken.french/data_library.html}). The weekly returns are respectively created from the daily 49 Industry Portfolio data set and the daily Factor data set. 

We also collect several state variables, which are lagged. The choice of state variables choice is based on \citet{adriancovar} and includes the variables: (1) equity volatility; (2) real estate return; (3) TED spread; (4) Treasury bill rate (change); (5) slope of the yield curve (change); (6) credit spread (change); (7) market return. Specifically, the equity volatility is calculated as the moving average of the most recent 22 daily squared market returns; the weekly real estate return and weekly market return are taken from the 49 Industry Portfolio and Factor data sets of Ken French, respectively; the TED spread is taken as the TEDRATE series  from the FRED database; the Treasury bill rate is taken as the secondary market rate on 3-month Treasury bills from the H.15 series from the FRED database; the slope of the yield curve is taken as the difference between the secondary market rate on 10-year Treasury securities and 3-month Treasury bills, both obtained from the H.15 series; the credit spread is taken as the difference between the Moody's Baa rate, taken as the BAA10Y series from the FRED database, and the secondary market rate on 10-year Treasury securities from the H.15 series.

Our sample runs from 11 January 2002 to 21 January 2022 for a total of 1,000 observations. We split the data into a training and test set, both consisting of 500 observations. We estimate the parameters of the models on the training set using the penalized methods outlined in this paper. The penalty parameters are selected using 5-fold cross validation. The folds are chosen as subsequent blocks of observations to account for the time-series features of the data. Our test set therefore runs from 2012 to 2022 and includes several crises, including the early stages of the Covid pandemic in 2020.

\subsection{High-dimensional features}\label{sec:dictionary}
To allow for a nonlinear relationships between the returns and the state variables, we create transformed state variables and returns using Chebyshev polynomials up to some degree $K$, which may be increasing with sample size $T$. In this case $p$ increases with the sample size $T$, such that we are in a high-dimensional setting. The theoretical results in this paper outline conditions under which the estimator is consistent and the prediction error is small in such settings.

We outline the procedure to obtain $\phi(R^{I}_t,Z_{t-1})$, denoting the regressor vector in the model of $R^M_t$. A similar procedure is applied to the regressor vector $\psi(Z_{t-1})$ in the model of $R^I_t$.
Specifically, $\phi(R^{I}_t,Z_{t-1}) = (1,\phi_{1,t},\ldots,\phi_{7K,t})'$, with $\phi_{(i-1)K+k,t} = \phi_k(S_{it})$ corresponding to the $k$§ degree Chebyshev polynomial of $S_{it}$, where $S_{it}$ is the $i$th element of $S_t = (R^I_t,Z_{t-1}')'$, and 
\begin{align*}
\phi_k(S_{it}) := \begin{cases}
\cos\left(k \arccos(\tilde{S}_{it})\right) & \text{ if } |\tilde{S}_{it}| \leq 1, \\
\cosh\left(k \, \text{arcosh} (\tilde{S}_{it})\right) & \text{ if } \tilde{S}_{it} > 1,\\
(-1)^k \cosh(k\, \text{arcosh}\left(-\tilde{S}_{it})\right) &  \text{ if } \tilde{S}_{it} < 1
\end{cases}
\end{align*}
with $\tilde{S}_{it} = (2S_{it}-a_i-b_i)/(b_i-a_i)$, where $a_i$ and $b_i$ are suitably chosen endpoints of the approximation interval $[a_i,b_i]$ for $S_{i}$. For simplicity, we choose $a_i = \min_{t\in[T]}(S_{it})$ and $b_i = \max_{t\in[T]}(S_{it})$.

\subsection{Estimation results}\label{sec:emp:results}

We generate VaR, ES, and CoES predictions for several choices of $K$, denoting the degree of Chebyshev polynomial transformations we consider. For $K=1$ we use the unpenalized estimator which is considered the benchmark. Indeed, \citet{adriancovar} use the unpenalized quantile regression estimator with similar state variables to obtain CoVaR estimates. We also consider $K= 2$, 3, 5, and 10 and use the penalized estimators for these values of $K$.

It is suspected that $K$ should be moderate for two reasons. First, for tail quantile levels like $\tau = 2.5\%$, Lemma \ref{lem:augregress-quantile-bound} shows that the estimation error in the quantile models, which depends on $K$, have a pronounced effect on the convergence rates of the ES predictions and estimator. Second, the time-series dependence in and heavy-tailedness of the data implies that the rate of $K$ is determined by the rates of $s_0$ and $p$ (see Condition \ref{cond:rates}) and worsens, \textit{ceteris paribus}, for higher dependence and heavy-tailedness. Financial data often are dependent over time and heavy-tailed, such that $K$ should be relatively small.

To compare VaR and ES predictions for different $K$, we use the following out-of-sample mean prediction errors. For VaR predictions the mean prediction error is captured by the mean tick-loss (MTL):
\[\text{MTL} = \sum_{t} \left(\tau - \mathds{1}(Y_t - \hat{Q}_t)\right) (Y_t - \hat{Q}_t),\]
where $\hat{Q}_t$ denotes a quantile prediction for period $t$ with parameters estimated over the training set and the sum runs over the periods in the test set.

For ES predictions the mean prediction error is captures by the mean squared prediction error for ES (ES-MSE):
\[\text{ES-MSE} = \sum_{t} \left(\hat{Y}_t - \widehat{ES}_t\right)^2,\]
where we recall $\hat{Y}_t = \hat{Q}_t + \frac{1}{\tau} \mathds{1}(Y_t < \hat{Q}_t) (Y_t - \hat{Q}_t)$, $\widehat{ES}_t$ denotes an ES prediction for period $t$ with parameters estimated over the training set, and the sum runs over the periods in the test set. The ES-MSE function is closely related to the objective function of the penalized ES estimator in this paper, see Equation (\ref{eq:ES-estimator}). The mean prediction errors MTL and ES-MSE will have different scales and should not be compared to each other.

Panel A in Table \ref{tab:emp-results} contains the mean prediction errors for the VaR and ES predictions of the weekly market return $R^M_t$ conditional on the weekly banking industry return $R^I_t$ and state variables $Z_{t-1}$. The results are shown for varying $K$, which denotes the degree of Chebyshev polynomials we consider. 

We observe that $K=3$ is the optimal choice, giving the smallest prediction errors for the VaR and ES predictions at quantile level $\tau =2.5\%$ and performing much better than the unpenalized benchmark ($K=1$). This provides good evidence in favor of  models of market risk that are nonlinear in the  industry return and state variables. Moreover, this value of $K$ may be considered moderate, such that it is reasonable the asymptotic convergence results of Lemma \ref{lem:asymp} apply.

At the median ($\tau=50\%$), the prediction error of the VaR prediction (median prediction) is less sensitive to $K$. The same holds to a lesser extent for the ES predictions. 
This can be explained by the \textit{mean-blur} property of high-frequency returns: the mean (median) of such returns is indistinguishable from zero due to low signal-to-noise ratio (see, e.g. \citet{christoffersen2011elements}).

Panel B in Table \ref{tab:emp-results} contains the mean prediction errors for the VaR predictions of the banking industry return $R^I_t$ conditional on the state variables $Z_{t-1}$. The choice $K=2$ is optimal, although its results are quite similar to $K=3$ at quantile level $\tau=2.5\%$. The mean-blur property, again, makes outcomes relatively insensitive to $K$ for the median ($\tau=50\%$).

Finally, Panel C presents average $\Delta$CoES estimates over the test set. The $\Delta$CoES estimates rely on each of the ES and VaR predictions we have discussed above, as evident from its definition in Section \ref{sec:coes}. As discussed above, the mean prediction error results in Panels A and B point towards $K=3$ as the optimal choice. At the optimal choice ($K=3$) we observe that average $\Delta$CoES is largest in absolute value. This suggests that the benchmark ($K=1$) underestimates systemic risk.

\begin{table}[htbp]
\footnotesize
  \centering
  \caption{Mean prediction errors (out-of-sample) \& average $\Delta-$CoVaR.}
    \begin{tabular}{rrrrrrr}
    \toprule
    \multicolumn{7}{c}{\boldmath{}\textbf{Panel A: VaR and ES prediction errors}\unboldmath{}} \\
    \multicolumn{7}{c}{(VaR \& ES of $R^{M}_t$ given $(R^I_t, Z_{t-1}')'$)} \\
    \midrule
          &       & \multicolumn{1}{l}{MTL} & \multicolumn{1}{l}{ES-MSE} &       & \multicolumn{1}{l}{MTL} & \multicolumn{1}{l}{ES-MSE} \\
\cmidrule{3-4}\cmidrule{6-7}    $K$   &       & \multicolumn{2}{c}{$\tau = 0.025$} &       & \multicolumn{2}{c}{$\tau = 0.5$ (median)} \\
\cmidrule{1-1}\cmidrule{3-4}\cmidrule{6-7}
    \multicolumn{1}{l}{1 (unpenalized)} &       & 0.084 & 93.774 &       & \textbf{0.186} & 0.557 \\
    2     &       & 0.038 & 12.706 &       & 0.188 & 0.394 \\
    3     &       & \textbf{0.032} & \textbf{6.840} &       & 0.191 & \textbf{0.387} \\
    5     &       & 0.040 & 21.735 &       & 0.186 & 0.446 \\
    10    &       & 0.046 & 30.383 &       & 0.199 & 0.421 \\
    \midrule
    \multicolumn{7}{c}{\boldmath{}\textbf{Panel B: VaR prediction errors}\unboldmath{}} \\
    \multicolumn{7}{c}{(VaR of $R^{I}_t$ given state variables $ Z_{t-1}$)} \\
    \midrule
          &       & \multicolumn{1}{l}{MTL} &       &       & \multicolumn{1}{l}{MTL} &  \\
\cmidrule{3-4}\cmidrule{6-7}    $K$   &       & \multicolumn{2}{c}{$\tau = 0.025$} &       & \multicolumn{2}{c}{$\tau = 0.5$ (median)} \\
\cmidrule{1-1}\cmidrule{3-4}\cmidrule{6-7}
    1 (unpenalized) &       & 0.101 &       &       & 0.492 &  \\
    2     &       & \textbf{0.095} &       &       & \textbf{0.484} &  \\
    3     &       & 0.096 &       &       & 0.486 &  \\
    5     &       & 0.123 &       &       & 0.485 &  \\
    10    &       & 0.099 &       &       & 0.484 &  \\
    \midrule
    \multicolumn{7}{c}{\textbf{Panel C: Average $\Delta$CoES ($\tau = 0.025$) }} \\
    \midrule
    $K$   &       & \multicolumn{1}{l}{} &       &       &       &  \\
\cmidrule{1-1} 
    1 (unpenalized) &       & -2.169 &       &       &       &  \\
    2     &       & -2.454 &       &       &       &  \\
    3     &       & -2.559 &       &       &       &  \\
    5     &       & -2.066 &       &       &       &  \\
    10    &       & -2.129 &       &       &       &  \\
    \bottomrule
    \end{tabular}%
  \label{tab:emp-results}%
  \begin{tablenotes}
  \item Note: This table contains out-of-sample mean prediction errors for VaR and ES forecasts (Panels A \& B) and the average $\Delta$CoES over the test set (Panel C). The prediction error is captured by the mean tick-loss (MTL) for VaR forecasts and the relevant squared error loss (ES-MSE) for ES forecasts. The parameter $K$ denotes the degree of Chebyshev polynomials we use to transform the each of the regressors (see Section \ref{sec:dictionary}). 
  \end{tablenotes}
\end{table}%

Finally, Figure \ref{fig:deltacor} show $\Delta$CoES predictions over the test set for $K=1$ and $K=3$, respectively denoted the small and large model. We observe that the large model responds quicker to news and predicts much greater risk spillover in several periods, including the first half of 2020, which corresponds to the onset of the Covid pandemic in the US. 

In Appendix \ref{appx:additional-plots} we include more plots of the ES and VaR predictions generated by the models discussed above for $K=1$ and 3 and $\tau=2.5\%$ and 50\%, as well as a plot of the market and industry returns. An interesting observation coming from these plots is that the larger ES and VaR models ($K=3$) for the market return are less sensitive to positive banking industry returns than the models that use $K=1$. This behavior is most pronounced for the large positive industry returns observed in 2020 and suggests that the inclusion of nonlinear transformations of the state variables can accommodate a \textit{leverage effect} for the choice $K=3$, i.e.~that negative industry returns increase market risk more than positive industry returns (see, e.g., \citet{christoffersen2011elements}). 

Finally, at $\tau = 50\%$ we note that the penalization of the VaR estimator for banking industry returns, gives predictions that are very close to zero for all periods. The unpenalized estimator ($K=1$) gives predictions that are much more volatile. This suggests that the penalized estimator is able to accommodate the low signal-to-noise ratio in the data as concerns the conditional median. 

\begin{figure}
\centering
\includegraphics[width=0.8\textwidth]{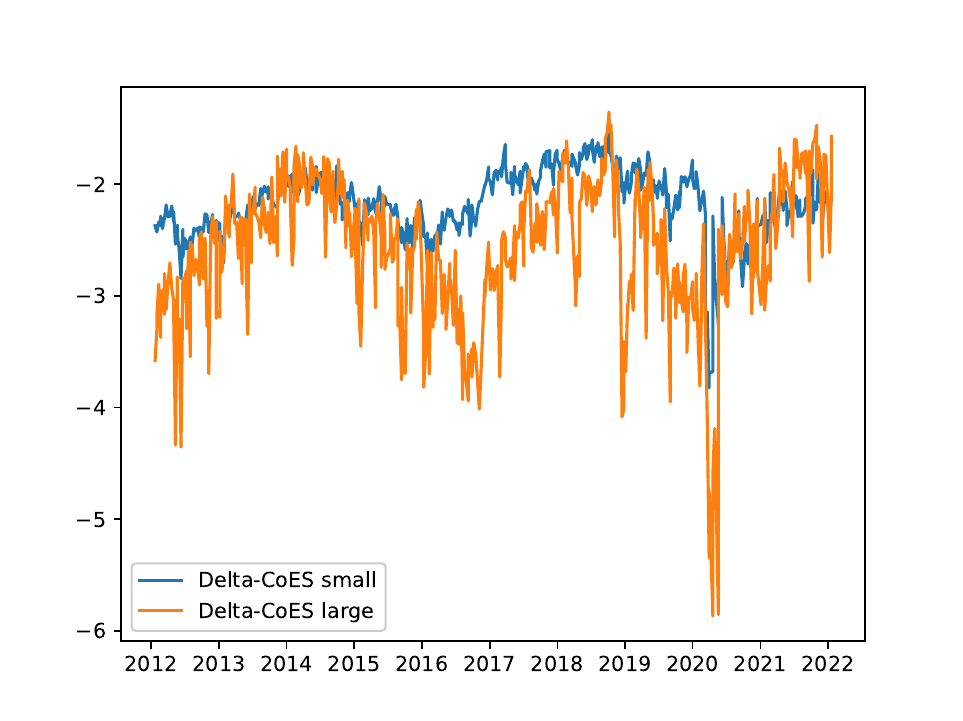}
\caption{$\Delta$CoES predictions. This plot shows predictions for $K=1$ (small) and $K=3$ (large) models.}
\label{fig:deltacor}
\end{figure}

\section{Simulation study}\label{sec:simulation}

\subsection{Simulation design}
We study the small sample properties of our estimator in a Monte Carlo simulation experiment. In each Monte Carlo iteration we estimate the parameters of the quantile and ES models on simulated data using our penalized ES estimator and the penalized QR estimator of BCMPW. We then compare the estimation and prediction errors of the estimators to those of their unpenalized counterparts. 

The simulation DGP we consider is a special case of (\ref{eq:linearmodel}):
\begin{align}
Y_t &= X_t'\xi + \left(X_t'\zeta\right) \nu_t,\label{eq:simDGP}
\end{align}
where $X_t$ denotes the regressor vector and $\nu_t\overset{iid}{\sim} N(0,\sigma_\nu^2)$ denotes the idiosyncratic error. In this model, $X_t'\xi$ and $X_t'\zeta$ describe the conditional mean and conditional volatility of $Y_t$, respectively. The elements of the regressor vector $X_t$ are generated as the Chebyshev polynomial transformations up to the $K$th degree of each of the elements of a regressor vector $Z_t$ to be defined below. We add one to each of Chebyshev polynomials and subsequently standardize by their respective standard deviations. The first manipulation ensures that the elements of $X_t$ are positive, such that $X_t'\zeta$ satisfies the positivity property of the conditional volatility if all elements of $\zeta$ are positive. The second manipulation ensures that the impact on $Y_t$ of the nonzero elements in the coefficient vectors $\xi$ and $\zeta$ is similar across Monte Carlo iterations. The regressor vector $X_t$ contains $p = 1 + dK$ elements, where $d$ denotes the dimension of the regressor vector $Z_t$, as it contains a constant and $K$ transformations of each of the elements in $Z_t$. That (\ref{eq:simDGP}) is a special case of (\ref{eq:linearmodel}) follows since we can rewrite $Y_t = X_t'\alpha(U_{t})$, where 
$\alpha(u) = \xi + \zeta Q_\eta(u)$ and $Q_\eta(u)$ denotes the quantile function corresponding to the distribution of $\nu_t$.

The regressor vector $Z_t$ is correlated in the cross-section and over time. Specifically, $Z_t = (Z_{1,t},\ldots,Z_{d,t})'$ and
\begin{align*}
Z_{i,t} &= \rho Z_{i,t-1} + \upsilon_{i,t},\\
\upsilon_{i,t} &= \sqrt{\theta} F_t + \sqrt{1-\theta^2} \psi_{i,t},
\end{align*}
where $F_t$ and $\psi_{i,t}$ are iid $N(0,1-\rho^2)$ distributed, such that $\text{Corr}\left(Z_{i,t},Z_{j,t}\right) = \theta$, $\text{Corr}\left(Z_{i,t},Z_{i,t-1}\right) = \rho$, and $\text{Var}(Z_{i,t}) = 1$, for all $i,j\in[d]$.

The coefficient vectors $\xi$ and $\zeta$ are chosen such that the vectors are sparse and contain $s_0$ nonzero elements each. The nonzero elements will differ between $\xi$ and $\zeta$. We consider $Z_{1,t}$, $Z_{2,t}$, and $Z_{3,t}$ to be relevant to the prediction of $Y_t$. Let the first $1+3K$ elements of $X_t$ be the constant and the transformations of $Z_{1,t}$, $Z_{2,t}$, and $Z_{3,t}$. In each Monte Carlo iteration we randomly draw the relevant transformations of $Z_{1,t}$, $Z_{2,t}$, and $Z_{3,t}$. Specifically, let $S_\xi \subset \{2,\ldots,1+3K\}$ and $S_\zeta \subset \{2,\ldots,1+3K\}$ denote the respective sparsity sets, with cardinality $s_0$ each, where the $s_0$ elements in each set are randomly drawn without replacement from the index set $\{2,\ldots,1+3K\}$. The sequence of draws is important in our design. Hence, we let $S_\xi^i$ and $S_\zeta^i$ denote the index that was drawn as the $i$th draw in the sequence. We then choose $\xi_{S_\xi^i} = 1/(1 + i)$ and $\xi_i = 0$ if $i \in [p] \setminus S_\xi$. Similarly, we choose $\zeta_{S_\zeta^i} = 1/(1 + i)$ and $\zeta_i = 0$ if $i \in [p] \setminus S_\zeta$.

We study the following parameter choices: cross-sectional correlation $\alpha = 0.15$; time-series correlation $\rho=0.5$; sample sizes $T=\{250,500\}$; amount of untransformed regressors $d = 7$; degree of Chebyshev transformations $K = \{3,5\}$; sparsity coefficient $s_0 = 2$; quantile levels $\tau = \{2.5\%,10\%\}$; and error variances $\sigma^2_\nu = \{1/4,1\}$. 
 As an example, in the scenario $d=7$, $K = 5$, and $T = 500$, the regression contains a constant and $dK = 35$ regressors, of which $2s_0 = 4$ are relevant. For both samples sizes and tail quantiles under consideration, we would not expect the unpenalized estimator to perform well in this scenario for weekly financial data.

\subsection{Simulation results}

For each parameter setting and each iteration we compute the estimation and prediction errors. The estimation error is measured by $\|\hat\gamma - \gamma^0\|_1$ for the ES estimator and by $\|\hat\alpha - \alpha^0\|_1$ for the quantile estimator. The prediction error is computed out-of-sample and given by the ES-MSE and MTL functions for the ES and quantile estimator, respectively, see Section \ref{sec:emp:results}. To obtain the out-of-sample prediction errors we draw a sample of $2T$ observations, reserve the first $T$ observations for estimation, and compute the out-of-sample prediction error on the final $T$ observations. For the penalized estimators, we choose the penalty parameter by five-fold cross-validation for time-series data and refer to Section \ref{sec:data-samplesplit-estimation} for details. We report the simulation average of the estimation and prediction errors over 1,000 Monte Carlo iterations.

Table \ref{tab:ES-sim} contains results for the ES estimator. We first discuss the simulation average of the estimation error in Panel A. In all scenarios we observe the penalized estimator outperforms the unpenalized counterpart. The penalized estimator improves if we increase the sample size from $T = 250$ to 500, which suggests our estimator is consistent. The unpenalized estimator improves with the sample size if $K=3$, but if $K=5$ its performance is extremely bad. This indicates that the performance of the unpenalized estimator already breaks down at a moderate amount of regressors. Finally, the performance of the estimator improves if the signal-to-noise ratio, as measured by $\sigma_\nu$, decreases or if the quantile level $\tau$ is increased.

A similar picture arises for the prediction error in Panel B in Table \ref{tab:ES-sim}. We specifically note that the simulation average of the  prediction error of the unpenalized estimator is about an order larger than that of the penalized estimator if $K = 5$. This suggests the out-of-sample performance of the predictions will be much better for the penalized estimator in comparison to the unpenalized estimator for scenarios with a moderate amount of regressors. We note that we should not compare the prediction errors at different quantile levels $\tau$, ceteris paribus, since the functional form of the prediction errors changes with $\tau$. Table \ref{tab:QR-sim} in the Appendix contains the simulation results for the quantile estimator. We generally find similar results, although at $K=5$ the difference between the penalized and unpenalized quantile estimator in terms of average prediction error is not as stark as for the ES estimator.

\begin{table}[htbp]
\footnotesize
  \centering
  \caption{Simulation results for ES estimator}
    \begin{tabular}{rlrrrrrr}
    \toprule
    \multicolumn{7}{c}{\textbf{Panel A: Average estimation error}}\\
    \midrule
          &       &       & \multicolumn{5}{c}{$K = 3$} \\
\cmidrule{4-8}          &       &       & \multicolumn{2}{c}{$\sigma_\nu = 1/4$} &       & \multicolumn{2}{c}{$\sigma_\nu = 1$} \\
\cmidrule{4-5}\cmidrule{7-8}          &       &       & $T =500$ & $1,000$ &       & $T =500$ & $1,000$ \\
\cmidrule{1-2}\cmidrule{4-5}\cmidrule{7-8}    \multicolumn{1}{l}{$\tau = 2.5\%$} &       &       &       &       &       &       &  \\
          & Penalized &       & 0.11  & 0.07  &       & 0.35  & 0.20 \\
          & Unpenalized &       & 0.43  & 0.36  &       & 2.29  & 1.82 \\
\cmidrule{1-2}\cmidrule{4-5}\cmidrule{7-8}    \multicolumn{1}{l}{$\tau = 10\%$} &       &       &       &       &       &       &  \\
          & Penalized &       & 0.04  & 0.03  &       & 0.13  & 0.11 \\
          & Unpenalized &       & 0.25  & 0.18  &       & 1.27  & 0.95 \\
    \midrule
          &       &       & \multicolumn{5}{c}{$K = 5$} \\
\cmidrule{4-8}          &       &       & \multicolumn{2}{c}{$\sigma_\nu = 1/4$} &       & \multicolumn{2}{c}{$\sigma_\nu = 1$} \\
\cmidrule{4-5}\cmidrule{7-8}          &       &       & $T =500$ & $1,000$ &       & $T =500$ & $1,000$ \\
\cmidrule{1-2}\cmidrule{4-5}\cmidrule{7-8}    \multicolumn{1}{l}{$\tau = 2.5\%$} &       &       &       &       &       &       &  \\
          & Penalized &       & 0.08  & 0.05  &       & 0.29  & 0.18 \\
          & Unpenalized &       & $>10^6$ & $>10^6$ &       & $>10^6$ & $>10^6$ \\
\cmidrule{1-2}\cmidrule{4-5}\cmidrule{7-8}    \multicolumn{1}{l}{$\tau = 10\%$} &       &       &       &       &       &       &  \\
          & Penalized &       & 0.02  & 0.02  &       & 0.10  & 0.08 \\
          & Unpenalized &       & $>10^6$ & $>10^6$ &       & $>10^6$ & $>10^6$ \\
\midrule
\multicolumn{7}{c}{\textbf{Panel B: Average prediction error}}\\
\midrule
      &       &       & \multicolumn{5}{c}{$K = 3$} \\
\cmidrule{4-8}      &       &       & \multicolumn{2}{c}{$\sigma_\nu = 1/4$} &       & \multicolumn{2}{c}{$\sigma_\nu = 1$} \\
\cmidrule{4-5}\cmidrule{7-8}      &       &       & $T =500$ & $1,000$ &       & $T =500$ & $1,000$ \\
\cmidrule{1-2}\cmidrule{4-5}\cmidrule{7-8}\multicolumn{1}{l}{$\tau = 2.5\%$} &       &       &       &       &       &       &  \\
      & Penalized &       & 25.41 & 11.08 &       & 127.33 & 62.14 \\
      & Unpenalized &       & 41.14 & 14.85 &       & 230.79 & 78.65 \\
\cmidrule{1-2}\cmidrule{4-5}\cmidrule{7-8}\multicolumn{1}{l}{$\tau = 10\%$} &       &       &       &       &       &       &  \\
      & Penalized &       & 2.75  & 2.29  &       & 15.22 & 13.53 \\
      & Unpenalized &       & 3.26  & 2.07  &       & 15.95 & 10.20 \\
\midrule
      &       &       & \multicolumn{5}{c}{$K = 5$} \\
\cmidrule{4-8}      &       &       & \multicolumn{2}{c}{$\sigma_\nu = 1/4$} &       & \multicolumn{2}{c}{$\sigma_\nu = 1$} \\
\cmidrule{4-5}\cmidrule{7-8}      &       &       & $T =500$ & $1,000$ &       & $T =500$ & $1,000$ \\
\cmidrule{1-2}\cmidrule{4-5}\cmidrule{7-8}\multicolumn{1}{l}{$\tau = 2.5\%$} &       &       &       &       &       &       &  \\
      & Penalized &       & 97.09 & 27.16 &       & 261.41 & 103.16 \\
      & Unpenalized &       & 873.05 & 172.78 &       & 4042.46 & 930.17 \\
\cmidrule{1-2}\cmidrule{4-5}\cmidrule{7-8}\multicolumn{1}{l}{$\tau = 10\%$} &       &       &       &       &       &       &  \\
      & Penalized &       & 3.04  & 2.38  &       & 15.74 & 13.26 \\
      & Unpenalized &       & 23.61 & 5.94  &       & 125.38 & 35.21 \\
\bottomrule
    \end{tabular}%
  \label{tab:ES-sim}%
  \begin{tablenotes}
  \item This table displays the simulation average of the estimation error of the ES estimator, as measured by $\|\hat\gamma - \gamma^0\|_1$, in Panel A and the simulation average of the prediction error of the ES estimator in Panel B. The prediction error is measured by the ES-MSE function (see Section \ref{sec:emp:results}). Results are included for the penalized and unpenalized estimator. The penalty is chosen by cross-validation in the penalized case and set to zero in the unpenalized case.
  \end{tablenotes}
\end{table}%

\section{Fuk-Nagaev inequality for beta-mixing}\label{sec:fuknagaev}
Finally, we derive a Fuk-Nagaev inequality for the maxima of high-dimensional sums of $\beta$-mixing sequences.
Before stating the Fuk-Nagaev inequality, we define the $\beta$-mixing property and the blocking strategy $(a_T,d_T)$. 

\begin{mydef}[$\beta$-mixing]\label{def:mixing}
Consider a random sequence $\{w_t\}$ on a probability space $(\Omega,\mathcal{F},P)$, and corresponding $\sigma$-fields $\mathcal{F}_{-\infty}^t := \sigma(\ldots,w_{t-1},w_{t})$ and $\mathcal{F}_{t}^\infty = \sigma(w_t,w_{t+1},\ldots)$. The sequence $\{w_t\}$ is $\beta$-mixing, if the mixing coefficient $\beta(k) := \sup_{t}\beta(\mathcal{F}_{-\infty}^t,\mathcal{F}_{t+k}^{\infty}) \rightarrow 0$ as $k\rightarrow \infty$, where, for any two $\sigma$-fields $\mathcal{G},\mathcal{H} \subset \mathcal{F}$,
\begin{align*}
\beta(\mathcal{G},\mathcal{H}) =\sup \frac{1}{2} \sum_{j=1}^L\sum_{j'=1}^{L'} \bigl| P(H_j\cap G_{j'}) - P(G_j)P(G_{j'})  \bigr|,
\end{align*}
with the supremum taken over all pairs of (finite) partitions $\{G_1,\ldots,G_L\}$ and $\{H_1,\ldots,H_{L'}\}$ of $\Omega$, such that $G_j \subset \mathcal{G}$ and $H_{j'}\subset \mathcal{H}$.
\end{mydef}

\begin{mydef}[Blocking strategy]\label{def:blockingstrat} Consider a pair $(a_T,d_T)$ of strictly positive integers satisfying $2a_Td_T \leq T$. For any sequence $\{w_1,\ldots,w_T\}$, we divide the sequence into $2d_T$ blocks of length $a_T$ and a remainder block of length $T-2d_T a_T$. Specifically, the blocks, as indicated by the time indices, are $H_j := \bigl\{t:2(j-1)a_T + 1 \leq t \leq (2j-1)a_T\bigr\}$, $G_j := \bigl\{t:(2j-1)a_T + 1 \leq t \leq 2j a_T\bigr\}$ and $Q := \bigl\{t:2j a_T +1  \leq t \leq T \bigr\}$, for $1\leq j \leq d_T$.
\end{mydef}

We are now ready to state the Fuk-Nagaev inequality.

\begin{mylem}\label{lem:fuknagaev}
Let $\{w_t\}$ be some $p$-dimensional strictly stationary and centered sequence that is $\beta$-mixing with coefficient $\beta(t)$. Moreover, for some $q \geq 2$, it holds  $\max_{1\leq i \leq p}E|w_{i1}|^q \leq \Delta_q$, where $w_{it}$ denotes the $i$th element of $w_t$. Let $(a_T,d_T)$ be a blocking strategy satisfying $T/2 - a_T \leq d_T a_T \leq T/2$.

Then, there exist uniform constants $C_q^{(1)},C_q^{(2)}>0$, such that, for any $u>0$,
\begin{align*}
P\left(\max_{1\leq i \leq p}\left|\frac{1}{T}\sum_{t=1}^T w_{t,i}\right| > u\right) 
&\leq 
3 p a_T \left(
\frac{C_q^{(1)}}{u^q d_T^{q-1}}
+ \exp\left(-C_q^{(2)}u^2d_T\right)
\right) + 
2 pd_T\beta(a_T).
\end{align*}
Moreover, there exist uniform constants $C_q^{(3)},C_q^{(4)}>0$, such that, for any $\delta \in (0,1)$, if $u = C_q^{(3)} \frac{(pa_T/\delta)^{1/q}}{d_T^{(q-1)/q}}\wedge C_q^{(4)} \frac{\sqrt{\log(pa_T/\delta)}}{\sqrt{d_T}}$, it follows
\begin{align*}
P\left(\max_{1\leq i \leq p}\left|\frac{1}{T}\sum_{t=1}^T w_{t,i}\right| > u\right) 
&\leq \delta + 2p d_T\beta({a_T}).
\end{align*}
\end{mylem}

\begin{proof}
See Appendix \ref{proof:fuknagaev}.
\end{proof}

\bibliography{ES_regression_WP.bib}

\appendix

\section{Appendix}

\subsection{Additional plots and simulation table}\label{appx:additional-plots}

\begin{figure}
\centering
\includegraphics[width=0.8\textwidth]{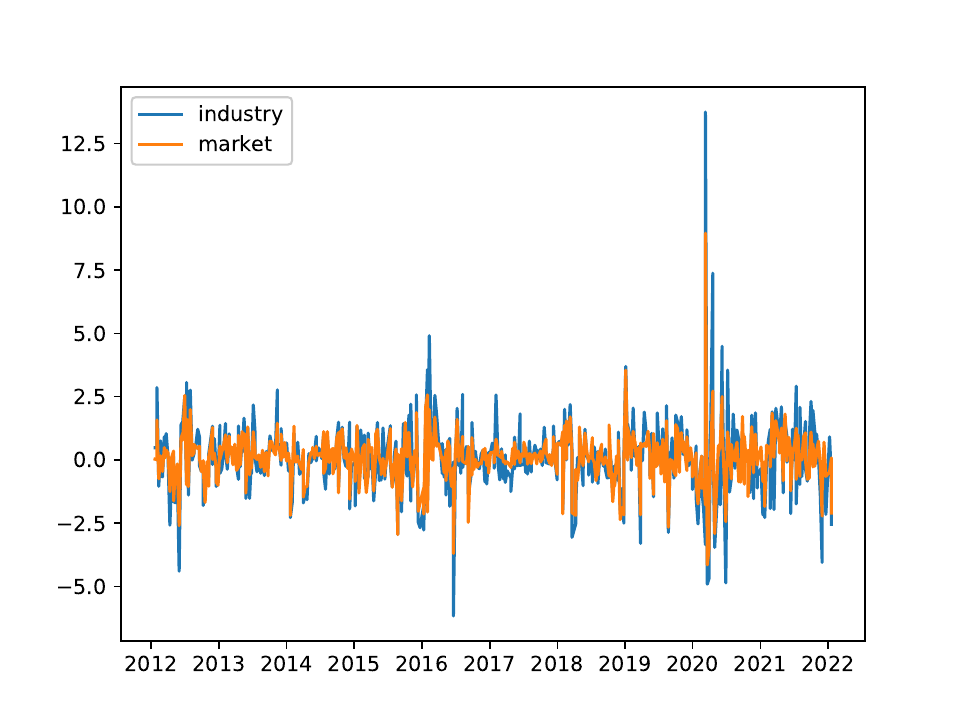}
\caption{Weekly market and industry return}
\label{fig:returns}
\end{figure}

\begin{figure}
\centering
\includegraphics[width=0.8\textwidth]{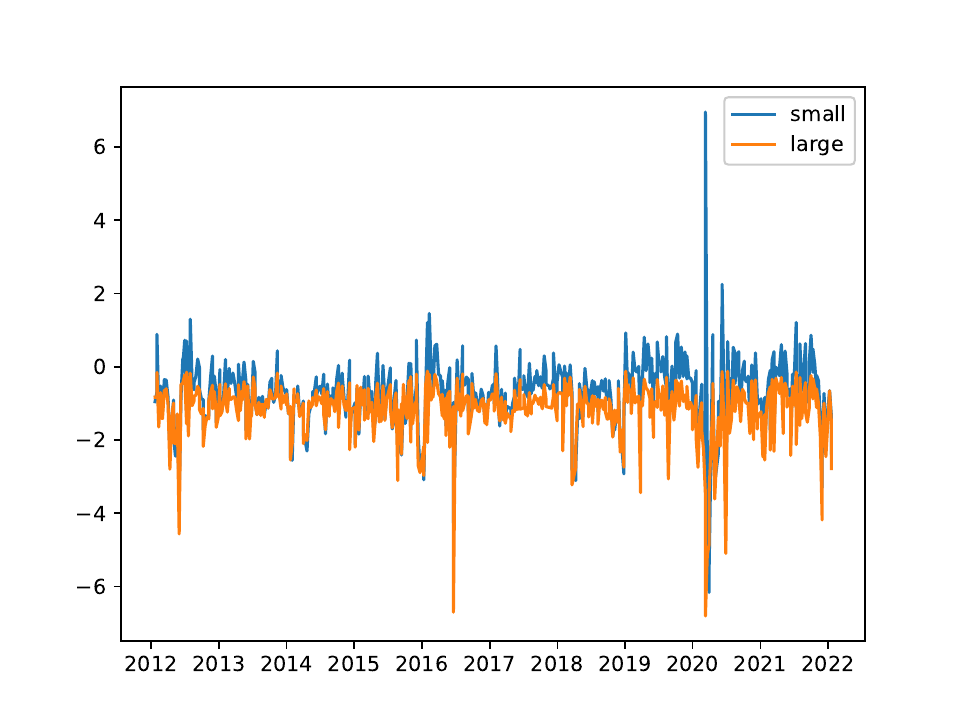}
\caption{ES predictions for $R^M_t$ conditional on $(R^I_t,Z_{t-1}')'$ at quantile level $\tau=0.025$. This plot shows predictions for $K=1$ (small) and $K=3$ (large) models.}
\label{fig:ES025}
\end{figure}

\begin{figure}
\centering
\includegraphics[width=0.8\textwidth]{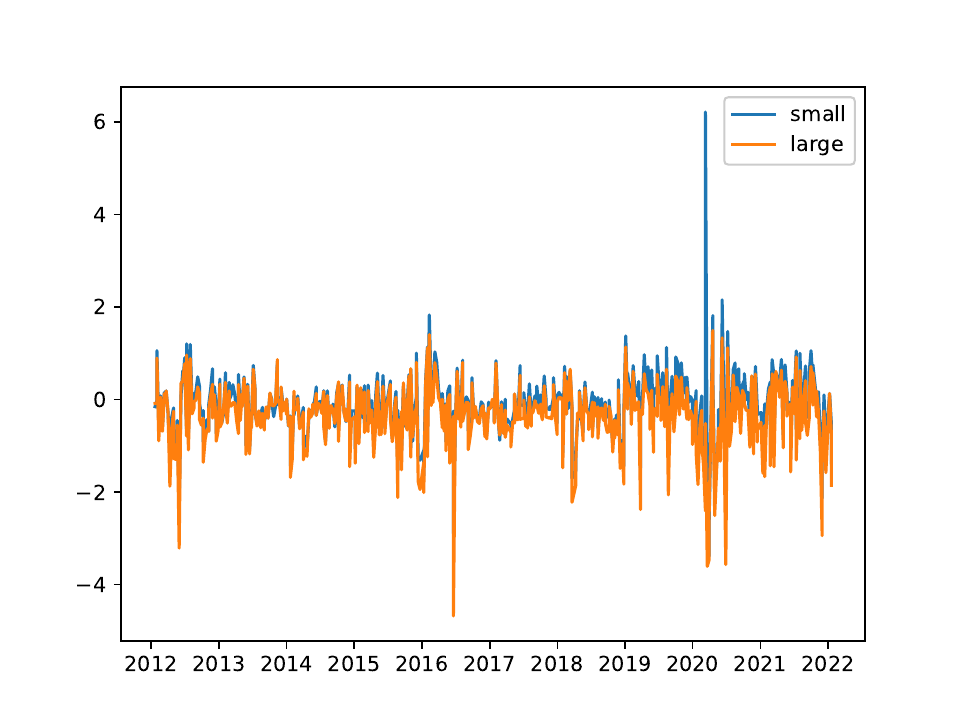}
\caption{ES predictions for $R^M_t$ conditional on $(R^I_t,Z_{t-1}')'$ at quantile level $\tau=0.5$. This plot shows predictions for $K=1$ (small) and $K=3$ (large) models.}
\label{fig:ES05}
\end{figure}

\begin{figure}
\centering
\includegraphics[width=0.8\textwidth]{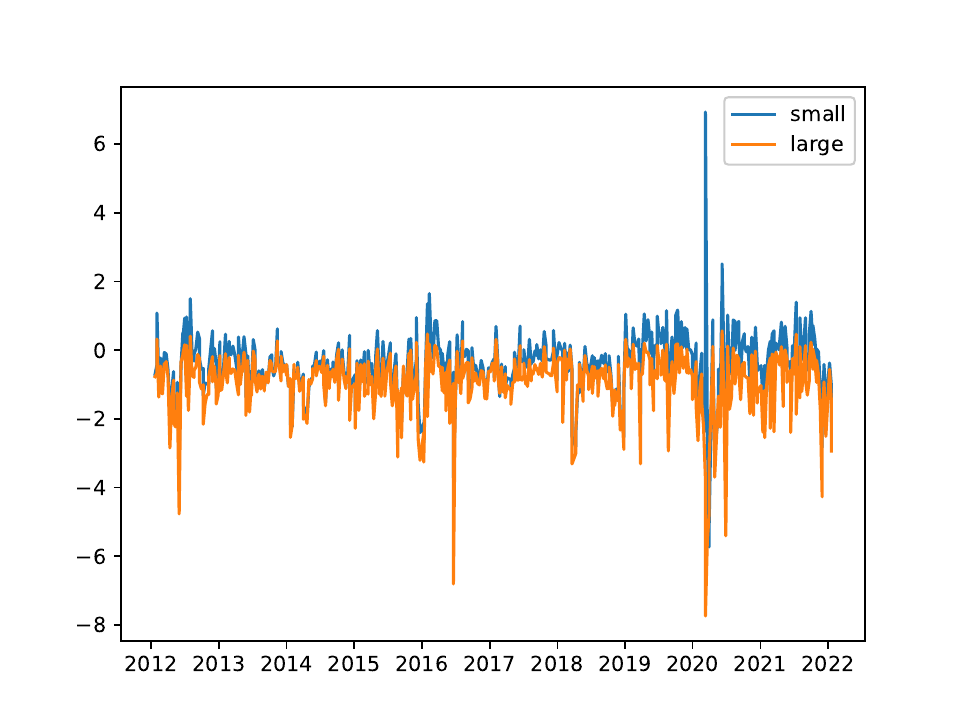}
\caption{VaR predictions for $R^M_t$ conditional on $(R^I_t,Z_{t-1}')'$ at quantile level $\tau=0.025$. This plot shows predictions for $K=1$ (small) and $K=3$ (large) models.}
\label{fig:Q025}
\end{figure}

\begin{figure}
\centering
\includegraphics[width=0.8\textwidth]{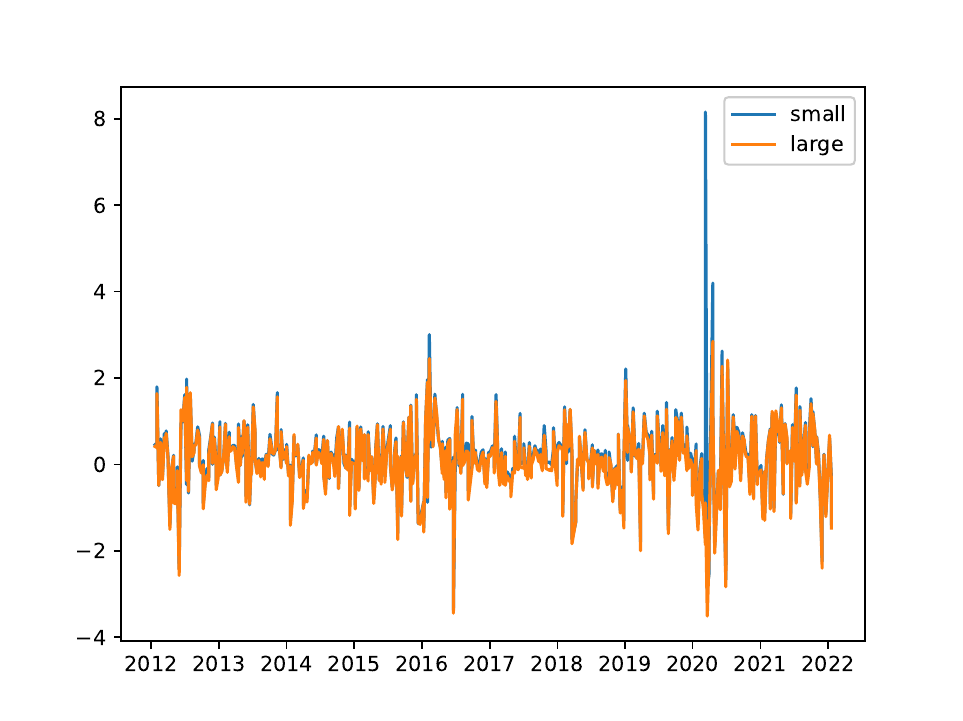}
\caption{VaR predictions for $R^M_t$ conditional on $(R^I_t,Z_{t-1}')'$ at quantile level $\tau=0.5$. This plot shows predictions for $K=1$ (small) and $K=3$ (large) models.}
\label{fig:Q05}
\end{figure}

\begin{figure}
\centering
\includegraphics[width=0.8\textwidth]{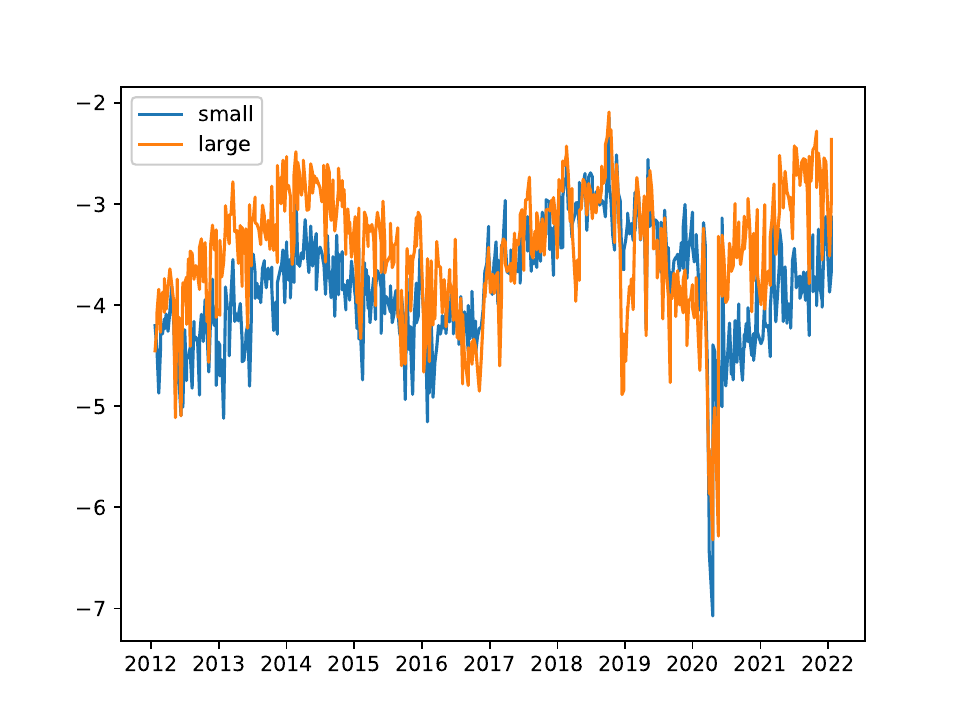}
\caption{VaR predictions for $R^I_t$ conditional on $Z_{t-1}$ at quantile level $\tau=0.025$. This plot shows predictions for $K=1$ (small) and $K=3$ (large) models.}
\label{fig:Q025-industry}
\end{figure}

\begin{figure}
\centering
\includegraphics[width=0.8\textwidth]{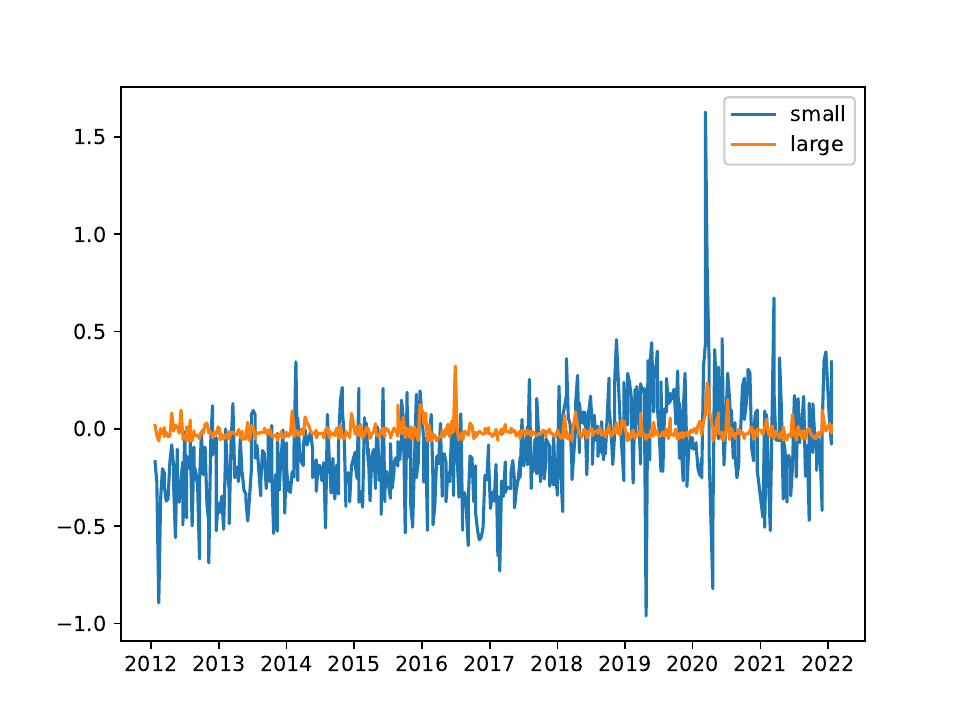}
\caption{VaR predictions for $R^I_t$ conditional on $Z_{t-1}$ at quantile level $\tau=0.5$. This plot shows predictions for $K=1$ (small) and $K=3$ (large) models.}
\label{fig:Q05-industry}
\end{figure}

\begin{table}[htbp]
\footnotesize
  \centering
  \caption{Simulation results for quantile estimator}
    \begin{tabular}{rlrrrrrr}
    \toprule
    \multicolumn{7}{c}{\textbf{Panel A: Average estimation error}}\\
    \midrule
          &       &       & \multicolumn{5}{c}{$K = 3$} \\
\cmidrule{4-8}          &       &       & \multicolumn{2}{c}{$\sigma_\nu = 1/4$} &       & \multicolumn{2}{c}{$\sigma_\nu = 1$} \\
\cmidrule{4-5}\cmidrule{7-8}          &       &       & $T =500$ & $1,000$ &       & $T =500$ & $1,000$ \\
\cmidrule{1-2}\cmidrule{4-5}\cmidrule{7-8}    \multicolumn{1}{l}{$\tau = 2.5\%$} &       &       &       &       &       &       &  \\
          & Penalized &       & 0.20  & 0.13  &       & 0.78  & 0.46 \\
          & Unpenalized &       & 0.51  & 0.38  &       & 2.73  & 1.97 \\
\cmidrule{1-2}\cmidrule{4-5}\cmidrule{7-8}    \multicolumn{1}{l}{$\tau = 10\%$} &       &       &       &       &       &       &  \\
          & Penalized &       & 0.06  & 0.05  &       & 0.15  & 0.10 \\
          & Unpenalized &       & 0.26  & 0.15  &       & 1.30  & 0.85 \\
    \midrule
          &       &       & \multicolumn{5}{c}{$K = 5$} \\
\cmidrule{4-8}          &       &       & \multicolumn{2}{c}{$\sigma_\nu = 1/4$} &       & \multicolumn{2}{c}{$\sigma_\nu = 1$} \\
\cmidrule{4-5}\cmidrule{7-8}          &       &       & $T =500$ & $1,000$ &       & $T =500$ & $1,000$ \\
\cmidrule{1-2}\cmidrule{4-5}\cmidrule{7-8}    \multicolumn{1}{l}{$\tau = 2.5\%$} &       &       &       &       &       &       &  \\
          & Penalized &       & 0.30  & 0.23  &       & 0.75  & 0.70 \\
          & Unpenalized &       & 3.53  & 2.32  &       & 15.68 & 9.59 \\
\cmidrule{1-2}\cmidrule{4-5}\cmidrule{7-8}    \multicolumn{1}{l}{$\tau = 10\%$} &       &       &       &       &       &       &  \\
          & Penalized &       & 0.04  & 0.04  &       & 0.12  & 0.12 \\
          & Unpenalized &       & 2.00  & 0.81  &       & 10.68 & 4.49 \\
\midrule
    \multicolumn{7}{c}{\textbf{Panel B: Average prediction error}}\\
    \midrule
      &       &       & \multicolumn{5}{c}{$K = 3$} \\
\cmidrule{4-8}      &       &       & \multicolumn{2}{c}{$\sigma_\nu = 1/4$} &       & \multicolumn{2}{c}{$\sigma_\nu = 1$} \\
\cmidrule{4-5}\cmidrule{7-8}      &       &       & $T =500$ & $1,000$ &       & $T =500$ & $1,000$ \\
\cmidrule{1-2}\cmidrule{4-5}\cmidrule{7-8}\multicolumn{1}{l}{$\tau = 2.5\%$} &       &       &       &       &       &       &  \\
      & Penalized &       & 0.04  & 0.04  &       & 0.10  & 0.08 \\
      & Unpenalized &       & 0.05  & 0.04  &       & 0.12  & 0.09 \\
\cmidrule{1-2}\cmidrule{4-5}\cmidrule{7-8}\multicolumn{1}{l}{$\tau = 10\%$} &       &       &       &       &       &       &  \\
      & Penalized &       & 0.10  & 0.10  &       & 0.23  & 0.23 \\
      & Unpenalized &       & 0.11  & 0.10  &       & 0.25  & 0.22 \\
\midrule
      &       &       & \multicolumn{5}{c}{$K = 5$} \\
\cmidrule{4-8}      &       &       & \multicolumn{2}{c}{$\sigma_\nu = 1/4$} &       & \multicolumn{2}{c}{$\sigma_\nu = 1$} \\
\cmidrule{4-5}\cmidrule{7-8}      &       &       & $T =500$ & $1,000$ &       & $T =500$ & $1,000$ \\
\cmidrule{1-2}\cmidrule{4-5}\cmidrule{7-8}\multicolumn{1}{l}{$\tau = 2.5\%$} &       &       &       &       &       &       &  \\
      & Penalized &       & 0.06  & 0.04  &       & 0.12  & 0.09 \\
      & Unpenalized &       & 0.12  & 0.06  &       & 0.26  & 0.14 \\
\cmidrule{1-2}\cmidrule{4-5}\cmidrule{7-8}\multicolumn{1}{l}{$\tau = 10\%$} &       &       &       &       &       &       &  \\
      & Penalized &       & 0.11  & 0.10  &       & 0.24  & 0.24 \\
      & Unpenalized &       & 0.15  & 0.12  &       & 0.35  & 0.26 \\
    \bottomrule
    \end{tabular}%
  \label{tab:QR-sim}%
    \begin{tablenotes}
  \item This table displays the simulation average of the estimation error of the quantile estimator, as measured by $\|\hat\alpha - \alpha^0\|_1$, in Panel A and the simulation average of the prediction error of the quantile estimator in Panel B. The prediction error is measured by the mean tick-loss (MTL) function (see Section \ref{sec:emp:results}). Results are included for the penalized and unpenalized estimator. The penalty is chosen by cross-validation in the penalized case and set to zero in the unpenalized case.
  \end{tablenotes}
\end{table}%

\subsection{On Condition \ref{cond:Variances}}\label{appx:conditionvariances}
For the case $\sigma_i = d_j$, $d_j\neq 1$, let $\tilde{X} = X D^{-1}$, where $D = \text{diag}(d_1,\ldots,d_p)$. The regressor matrix $\tilde{X}$ does satisfy Condition \ref{cond:Variances} and the equivalent model to Equation (\ref{eq:linearmodel}) is $Y_t = \tilde{X}_t' D \alpha^0(U_t)$. Consequently, $ES_t(\tau) = \tilde{X}_t' \tilde\gamma^0(\tau)$, where $\tilde\gamma^0(\tau) = D\gamma_0(\tau)$. Since 
\begin{align*}
\frac{1}{n}\|\hat{Y} - X\gamma\|_2^2 + \lambda \sum_{j=1}^p \hat\sigma_i |\gamma_i| = \frac{1}{n}\|\hat{Y} - \tilde{X} \tilde\gamma\|_2^2 + \lambda \sum_{j=1}^p \tilde{\sigma}_j|\tilde{\gamma}_j|,
\end{align*}
where $\tilde{\sigma}_j^2 := \frac{1}{T}\sum_{t=1}^T \tilde{X}_{jt}^2 = \frac{\hat\sigma_i^2}{d_j^2}$, it follows that 
\[\hat{\tilde{\gamma}}:= \argmin_{\tilde{\gamma}} \frac{1}{n}\|\hat{Y} - \tilde{X}\tilde{\gamma}\|_2^2 + \lambda \sum_{j=1}^p \tilde{\sigma}_j|\tilde{\gamma}_j| = D\hat{\gamma}.\]
Moreover, $\hat{\tilde{\gamma}}$ satisfies the bounds established in the paper. Hence, the estimation error in the alternative case satisfies $\|\hat{\gamma} - \gamma_0\|_1 = \|D^{-1}(\hat{\tilde{\gamma}} - \tilde{\gamma}_0)\|_1 \leq 
\|D^{-1}\|_\infty \|\hat{\tilde{\gamma}} - \tilde{\gamma}_0\|_1$.

\subsection{Conditions of BCMPW}\label{appx:belloni}
Consider the DGP in (\ref{eq:linearmodel}). We impose the following assumptions.

\begin{customcond}{B1}\label{cond:B1}
There exists a $\underline{f}\in(0,1]$ satisfying
\[\inf_{x\in\mathcal{X}} f_{Y_t | X_t}\left(x'\alpha^0|x\right) > \underline{f},\]
where $f_{Y_t | X_t}$ denotes the density of $Y_t$ conditional on $X_t$. Furthermore, $f_{Y_t | X_t}\left(y|x\right)$ and $\frac{\partial}{\partial y}f_{Y_t | X_t}\left(y|x\right)$ are both bounded by $\bar{f}$ and $\bar{f}'$, respectively, uniformly in $y$ and $x$ in their respective supports.
\end{customcond}


As an example of a conditional distribution of $Y_t$ given $X_t$, consider the location-scale model $Y_t = E[Y_t|X_t] + \text{std}(Y_t|X_t) \nu_t$, where $E[Y_t|X_t] = X_t'\xi$ and $\text{std}(Y_t|X_t) = X_t'\zeta$, with location and scale parameter vectors $\xi \in \mathbb{R}^p$ and $\zeta \in \mathbb{R}^p$, respectively. We impose positive standard deviations and independent errors. Specifically, (i) $\inf_{x\in\mathcal{X}} x'\zeta > 0$; (ii) $\nu_t \sim iid \,\, F_\nu$ for some distribution $F_\nu$ with continuously differentiable cdf and $\nu_t$ independent of $X_t$ and its past. Clearly, this model can be expressed in terms of the model in Equation (\ref{eq:linearmodel}) as follows: $Y_t = X_t'\alpha^0(U_t)$, where $\alpha^0(u) = \xi + F_\nu^{-1}(u)\zeta$, with $F_\nu^{-1}(\cdot)$ denoting the inverse cdf. It follows the conditional pdf can be expressed as $f_{Y_t|X_t}(y|x) = \frac{1}{x'\zeta} f_\nu\left((y - x'\xi) / (x'\zeta) \right)$. 

We let $S_{\alpha,0}$ denote the sparsity set of $\alpha^0$, and let $s_{\alpha,0}$ denote the cardinality of $S_{\alpha,0}$. Here, we specifically allow for different sparsity sets for the quantile and ES model parameters, respectively, $\alpha^0$ and $\gamma^0$. Moreover, we define the cone $A^\alpha = \{\alpha : \|\alpha_{S_{\alpha,0}^c}\|_1 \leq C_0 \|\alpha_{S_{\alpha,0}}\|_1\}$, for some positive constant $C_0$.

\begin{customcond}{B2}\label{cond:B2}
The following conditions are satisfied.
\begin{enumerate}
  \item 
For any $\alpha\in A^\alpha$, let
\[J_\tau(\alpha) := \underline{f} E\left[ \left(X_t'\alpha\right)^2 \right] = \underline{f} \alpha'\Sigma\alpha.\]
It holds 
\[0 < \kappa_0:= \inf_{\alpha \in A^\alpha,\alpha \neq 0}\frac{J_\tau^{1/2}(\alpha)}{\left\|\alpha_{S_{\alpha,0}}\right\|_2};
\]
\item
Moreover, we assume the following holds
\[0 < q:= \frac{3}{8}\frac{\underline{f}^{3/2}}{\bar{f}'} \inf_{\alpha\in A^\alpha,\alpha\neq 0} \frac{E\left(\sum_{t=1}^T (X_t'\alpha)^2 \right)^{3/2}}{E\left(\sum_{t=1}^T |X_t'\alpha|^3 \right)}.\]
\end{enumerate}
\end{customcond}

\begin{customcond}{B3}\label{cond:belloni:rates}
Suppose
\
\begin{enumerate}
  \item $q \geq \frac{C}{{\kappa_0 \underline{f}^{1/2}}} (\psi_T\sqrt{a_T \log(p a_T)}\sqrt{s_{\alpha,0} + 1})$, for some universal constant $C>0$ and any sequence $\{\psi_T\}$ that satisfies \[\psi_T / \left(\sqrt{\underline{f}} \log(a_T+1)\right)\rightarrow \infty;\]
  \item $\frac{\psi_T \sqrt{(1+s_{\alpha,0})a_T\log(pa_T)}}{\sqrt{d_T}\kappa_0 \underline{f}} = o(1)$;
    \item $\frac{\underline{f} s_{\alpha,0}}{\kappa_0^2} \frac{(pa_T)^{2/q}}{d_T^{(q-1)/q}} = o(1)$. 
  \end{enumerate}
\end{customcond}

We note that Conditions \ref{cond:B2}.1 and \ref{cond:belloni:rates} are closely related to \ref{cond:REC} and \ref{cond:rates}. Indeed, Condition \ref{cond:B2}.1 can be rewritten as $\|\alpha_{S_{\alpha,0}}\|_1^2 \leq s_{\alpha,0}\|\alpha_{S_{\alpha,0}}\|_2^2 \leq \frac{s_{\alpha,0}\underline{f}}{\kappa_0^2} \alpha'\Sigma \alpha$, for any $\alpha \in A_\alpha$, which is similar to Condition \ref{cond:REC}. Moreover, if we assume $\kappa_0$ and $\underline{f}$ are universal constants (like $\phi_0$ in \ref{cond:REC}), then the rates in \ref{cond:belloni:rates}.2 and \ref{cond:belloni:rates}.3 no longer depend on $\kappa_0$ and $\underline{f}$, and are similar to those in Condition \ref{cond:rates} (ignoring $\psi_T$). Condition \ref{cond:B2}.2 represents the restricted nonlinear impact condition introduced in \citet{belloni2011} and is standard in the $\ell_1$-penalized quantile regression literature.

\subsection{Lemma \ref{lem:nonasymp-prob}}
In this section we allow for a more general penalty level 
\[\lambda \asymp \frac{(pa_T/\delta_T)^{1/q}}{d_T^{(q-1)/q}}\vee \frac{\sqrt{\log(pa_T/\delta_T)}}{\sqrt{d_T}},\]
 where $\{\delta_T\}$ is some deterministic sequence of constants on the interval $(0,1)$. In the main text we consider the special case $\delta_T = 1/(pa_T)$.

\begin{mylem}\label{lem:nonasymp-prob}
Let Conditions \ref{cond:Variances}, \ref{cond:REC}, \ref{cond:mixing}, and \ref{cond:moments} be satisfied.  Let $(a_T,d_T)$ be a blocking strategy satisfying $T/2 - a_T \leq d_T a_T \leq T/2$. If $\lambda_0 \leq \frac{6C_0 - 10}{15(1+C_0)} \lambda$ and $\lambda_1 \leq \frac{\phi_0^2}{2s_0 (1+C_0)^2}$, there exist universal constants $C_1,C_2,C_3>0$, such that 
  \[P(\mathcal{T}  \cap \mathcal{S} \cap \mathcal{U}) = 1 - A_\mathcal{S} - C_1\delta_T
  - A_\mathcal{U} - 4pd_T\bar\beta(a_T) - 2p^2d_T\bar\beta(a_T),\]
  where 
  \begin{align*}
  A_\mathcal{S} &:= 3 p a_T \Biggl(
\frac{C_2}{d_T^{q-1}}
+  \exp\bigl(-C_3 d_T\bigr)
\Biggr)\ \ \text{ and } \ \ 
A_\mathcal{U} :=
3 p^2 a_T \Biggl(
\frac{C_2 s_0^q}{ d_T^{q-1}}
+  \exp\bigl(-C_3 d_T/s_0^2\bigr)
\Biggr).
  \end{align*}
\end{mylem}

\begin{proof}
See Appendix \ref{proof:nonasymp-prob}.
\end{proof}

\subsection{Proof of Lemma \ref{lem:nonasymp-bound}}\label{proof:nonasymp-bound}

We have the Basic Inequality
\[\frac{1}{T}\|\hat{Y} - X\hat{\gamma}\|_2^2 +  \lambda \|\hat{\gamma}\|_{1,T} \leq 
\frac{1}{T}\|\hat{Y} - X\gamma^0\|_2^2 +  \lambda  \|\gamma^0\|_{1,T},\]
and rewrite
\begin{align*}
&\frac{1}{T}\|\hat{Y} -\tilde{Y} + \varepsilon - X(\hat{\gamma}-\gamma^0) \|_2^2 +  \lambda \|\hat{\gamma}\|_{1,T} \leq
\frac{1}{T}\|\hat{Y} -\tilde{Y} + \varepsilon\|_2^2 +  \lambda  \|\gamma^0\|_{1,T}.
\end{align*}

We expand both sides to obtain
\begin{align*}
&\frac{1}{T}\|X(\hat{\gamma}-\gamma^0) \|_2^2  + \frac{2}{T}\varepsilon'X(\hat{\gamma}-\gamma^0) + \frac{2}{T} (\hat{Y} -\tilde{Y})'X(\hat{\gamma}-\gamma^0) +  \lambda \|\hat{\gamma}\|_{1,T} \leq  \lambda  \|\gamma^0\|_{1,T}\\
\implies\\
&\frac{1}{T}\|X(\hat{\gamma}-\gamma^0) \|_2^2  +  \lambda \|\hat{\gamma}\|_{1,T} \leq \max_{1\leq j \leq p} \frac{2}{T}\left|\varepsilon'X_j \right|\|\hat{\gamma}-\gamma^0\|_1  + \frac{2}{\sqrt{T}}\|\hat{Y} -\tilde{Y}\|_2\frac{1}{\sqrt{T}}\|X(\hat{\gamma}-\gamma^0)\|_2  + \lambda  \|\gamma^0\|_{1,T}.
\end{align*}

Let $\lambda_2 := 2 \|\hat{Y} - \tilde{Y}\|_2/\sqrt{T}$.

On $\mathcal{T}$, and by the Triangle Inequality,
\begin{align*}
&\frac{1}{T}\|X(\hat{\gamma}-\gamma^0) \|_2^2  + \lambda \|\hat{\gamma}-\gamma^0\|_{1,T} \leq ( \lambda+\lambda_0) \|\hat{\gamma}-\gamma^0\|_1  + \lambda_2 \frac{1}{\sqrt{T}} \|X(\hat{\gamma}-\gamma^0)\|_2  + \lambda \left( \|\gamma^0\|_{1,T} - \|\hat{\gamma}\|_{1,T} \right).
\end{align*}

On $\mathcal{S}$, and by the Triangle Inequality, the right hand side is upper-bounded by
\begin{align*}
&(\lambda +\lambda_0)\|\hat{\gamma}-\gamma^0\|_1  + \lambda_2 \frac{1}{\sqrt{T}} \|X(\hat{\gamma}-\gamma^0)\|_2  + \lambda \|\gamma^0-\hat{\gamma}\|_{1,T}\\
&\quad \leq
(\lambda+\lambda_0) \|\hat{\gamma}-\gamma^0\|_1  + \lambda_2 \frac{1}{\sqrt{T}} \|X(\hat{\gamma}-\gamma^0)\|_2  + \frac{4}{3}\lambda \|\hat{\gamma}-\gamma^0\|_1\\
&\quad =
 \left(\frac{7}{3}\lambda + \lambda_0\right) \|\hat{\gamma}-\gamma^0\|_1  + \lambda_2 \frac{1}{\sqrt{T}} \|X(\hat{\gamma}-\gamma^0)\|_2
\end{align*}
and $\frac{1}{T}\|X(\hat{\gamma}-\gamma^0) \|_2^2  + \frac{4}{5} \lambda \|\hat{\gamma}-\gamma^0\|_1$ is upper-bounded by the LHS,
where we have used $\frac{3}{4}\|\gamma\|_{1,T} \leq \|\gamma\|_{1} \leq \frac{5}{4}\|\gamma\|_{1,T}$, for any $\gamma\in \mathbb{R}^p$, by Condition \ref{cond:Variances}. We obtain
\[\frac{1}{T}\|X(\hat{\gamma}-\gamma^0) \|_2^2  + \frac{4}{5} \lambda \|\hat{\gamma}-\gamma^0\|_1 \leq  \left(\frac{7}{3}\lambda + \lambda_0\right) \|\hat{\gamma}-\gamma^0\|_1  + \lambda_2 \frac{1}{\sqrt{T}} \|X(\hat{\gamma}-\gamma^0)\|_2.\]

We consider three cases and suppose that $\hat\gamma - \gamma^0 \in A$ in Cases 1 and 3. This condition will be established in the final part of the proof.

Case 1: $\lambda_2 \frac{1}{\sqrt{T}} \|X(\hat{\gamma}-\gamma^0)\|_2 \leq \lambda_0  \|\hat{\gamma}-\gamma^0\|_1$. Then, 
\begin{align*}
\frac{1}{T}\|X(\hat{\gamma}-\gamma^0) \|_2^2  + \frac{4}{5} \lambda \|\hat{\gamma}-\gamma^0\|_1 
&\leq 
\left(\frac{7}{3}\lambda + 2\lambda_0\right) \|\hat{\gamma}-\gamma^0\|_1,
\end{align*}
where the RHS is bounded by
\[
\left(\frac{7}{3}\lambda + 2\lambda_0\right) (1 + C_0) \frac{\sqrt{2s_0}}{\phi_0} \|X(\hat{\gamma}-\gamma^0)\|_2,
\]
by the restricted eigenvalue condition (\ref{cond:REC}) on $\mathcal{U}$.

Indeed, for any $\gamma \in A$, under the condition REC and on $\mathcal{U}$,
\begin{align*}
\left|\frac{\gamma'\hat\Sigma\gamma}{\gamma'\Sigma\gamma} - 1\right|
\leq 
\frac{1}{\gamma'\Sigma\gamma} \left|\gamma'(\hat\Sigma - \Sigma)\gamma \right|
\leq 
\frac{1}{\gamma'\Sigma\gamma} \left\|\gamma\right\|_1^2 \|\hat\Sigma - \Sigma\|_\infty
\leq 
\frac{(1+C_0)^2s_0}{\phi_0^2} \lambda_1 \leq \frac{1}{2},
\end{align*}
such that $\frac{2}{3} \gamma'\hat\Sigma\gamma \leq \gamma'\Sigma\gamma \leq 2 \gamma'\hat\Sigma\gamma$. 

Finally, utilizing in the second inequality that $2uv \leq u^2 + v^2$ for any reals $u$ and $v$, we obtain
\begin{align*}
\frac{1}{T}\|X(\hat{\gamma}-\gamma^0) \|_2^2  + \frac{4}{5} \lambda \|\hat{\gamma}-\gamma^0\|_1 
\leq \frac{1}{T}\|X(\hat{\gamma}-\gamma^0) \|_2^2  + \frac{8}{5} \lambda \|\hat{\gamma}-\gamma^0\|_1 
\leq 2 \Bigl(\frac{7}{3}\lambda +2\lambda_0\Bigr)^2(1 + C_0)^2 \frac{s_0}{\phi_0^2}.
\end{align*}

Case 2: $\lambda_2 \frac{1}{\sqrt{T}} \|X(\hat{\gamma}-\gamma^0)\|_2 > \lambda_0 \|\hat{\gamma}-\gamma^0\|_1$ and $2\lambda_2 \geq \frac{1}{\sqrt{T}} \|X(\hat{\gamma}-\gamma^0)\|_2$. Then,
\begin{align*}
\frac{1}{T}\|X(\hat{\gamma}-\gamma^0) \|_2^2  + \frac{4}{5} \lambda \|\hat{\gamma}-\gamma^0\|_1
&\leq \left(\frac{7}{3}\lambda + \lambda_0\right) \|\hat{\gamma}-\gamma^0\|_1  + \lambda_2 \frac{1}{\sqrt{T}} \|X(\hat{\gamma}-\gamma^0)\|_2\\ 
&\leq \left(\left(\frac{7}{3}\lambda + \lambda_0\right)/\lambda_0 + 1 \right) \lambda_2 \frac{1}{\sqrt{T}} \|X(\hat{\gamma}-\gamma^0)\|_2 \leq \left(\left(\frac{7}{3}\lambda + \lambda_0\right)/\lambda_0 + 1 \right) \lambda_2^2.
\end{align*}

Case 3: $\lambda_2 \frac{1}{\sqrt{T}} \|X(\hat{\gamma}-\gamma^0)\|_2 > \lambda_0 \|\hat{\gamma}-\gamma^0\|_1$ and $2 \lambda_2 < \frac{1}{\sqrt{T}} \|X(\hat{\gamma}-\gamma^0)\|_2$. Then,
\begin{align*}
\frac{1}{T}\|X(\hat{\gamma}-\gamma^0) \|_2^2  + \lambda \|\hat{\gamma}-\gamma^0\|_1 
&\leq \left(\frac{7}{3}\lambda + \lambda_0\right) \|\hat{\gamma}-\gamma^0\|_1  + \lambda_2 \frac{1}{\sqrt{T}} \|X(\hat{\gamma}-\gamma^0)\|_2\\
&\leq \left(\frac{7}{3}\lambda + \lambda_0\right) \|\hat{\gamma}-\gamma^0\|_1  + \frac{1}{2}\frac{1}{T} \|X(\hat{\gamma}-\gamma^0)\|_2^2.
\end{align*}
Hence, under Condition \ref{cond:REC} and on $\mathcal{U}$,
\begin{align*}
\frac{1}{T}\|X(\hat{\gamma}-\gamma^0) \|_2^2  + 2 \lambda \|\hat{\gamma}-\gamma^0\|_1 
&\leq 2 \left(\frac{7}{3}\lambda + \lambda_0\right) \|\hat{\gamma}-\gamma^0\|_1\\
&\leq 
\frac{1}{2}\Biggl(8\left(\frac{7}{3}\lambda + \lambda_0\right)^2(1 + C_0)^2 \frac{s_0}{\phi_0^2} + \frac{1}{T}\|X(\hat{\gamma}-\gamma^0)\|_2^2\Biggr),
\end{align*}
such that
\begin{align*}
\frac{1}{T}\|X(\hat{\gamma}-\gamma^0) \|_2^2  + 4 \lambda \|\hat{\gamma}-\gamma^0\|_1 
&\leq 8\left(\frac{7}{3}\lambda + \lambda_0\right)^2(1 + C_0)^2 \frac{s_0}{\phi_0^2}.
\end{align*}

Finally, we show that $\hat\gamma - \gamma^0 \in A$ in Cases 1 and 3.

Case 1: $\lambda_2 \frac{1}{\sqrt{T}} \|X(\hat{\gamma}-\gamma^0)\|_2 \leq \lambda_0  \|\hat{\gamma}-\gamma^0\|_1$. From
the Basic Inequality
\begin{align*}
  &\frac{1}{T}\|X(\hat{\gamma}-\gamma^0) \|_2^2  +  \lambda \|\hat{\gamma}\|_{1,T} \leq 
  \lambda_0 \|\hat{\gamma}-\gamma^0\|_1  + \lambda_2 \frac{1}{\sqrt{T}}\|X(\hat{\gamma}-\gamma^0)\|_2  + \lambda  \|\gamma^0\|_{1,T}.
\end{align*}
Using the identity $\|\hat{\gamma}\|_{1,T} = \|\hat{\gamma}_{S^c}\|_{1,T} + \|\hat{\gamma}_S\|_{1,T}$, and inequalities $\|\hat{\gamma}\|_{1,T} - \|\gamma^0\|_{1,T} \leq \|\hat{\gamma}- \gamma^0\|_{1,T}$ and 
$\frac{3}{4}\|\gamma\|_{1,T} \leq \|\gamma\|_{1} \leq \frac{5}{4}\|\gamma\|_{1,T}$, where the former follows by the Triangle Inequality and the latter by Condition \ref{cond:Variances}, we obtain
\begin{align*}
  \frac{1}{T}\|X(\hat{\gamma}-\gamma^0) \|_2^2  +  \lambda \|\hat{\gamma}_{S^c}\|_{1,T} &\leq 
  2\lambda_0 \|\hat{\gamma}-\gamma^0\|_1 + \lambda  \|\hat{\gamma}_{S}-\gamma^0_{S}\|_{1,T} \\
  &\implies
  \\
  \frac{4}{5} \lambda \|\hat{\gamma}_{S^c}\|_1 &\leq 
   \left(\frac{4}{3}\lambda+2\lambda_0\right) \|\hat{\gamma}_S - \gamma^0_S \|_1 + 2\lambda_0 \|\hat{\gamma}_{S^c}\|_1.
\end{align*}
Hence,
\begin{align*}
  \|\hat{\gamma}_{S^c}\|_{1} &\leq 
\frac{\frac{4}{3}\lambda +2\lambda_0}{\frac{4}{5}\lambda - 2\lambda_0} \|\hat{\gamma}_S - \gamma^0_S \|_1.
\end{align*}
Case 3:
$\lambda_2 \frac{1}{\sqrt{T}} \|X(\hat{\gamma}-\gamma^0)\|_2 > \lambda_0 \|\hat{\gamma}-\gamma^0\|_1$ and $2 \lambda_2 < \frac{1}{\sqrt{T}} \|X(\hat{\gamma}-\gamma^0)\|_2$. Then,
\begin{align*}
  &\frac{1}{2T}\|X(\hat{\gamma}-\gamma^0) \|_2^2  +  \lambda \|\hat{\gamma}\|_{1,T} \leq 
  \lambda_0 \|\hat{\gamma}-\gamma^0\|_1  + \lambda  \|\gamma^0\|_{1,T}.
\end{align*}
Hence,
\begin{align*}
  &\lambda \|\hat{\gamma}_{S^c}\|_{1,T} \leq 
  \lambda_0 \|\hat{\gamma}-\gamma^0\|_1  + \lambda  (\|\gamma^0_S\|_{1,T} -  \|\hat{\gamma}_S\|_{1,T})\\
  \implies\\
    &\frac{4}{5}\lambda \|\hat{\gamma}_{S^c}\|_{1} \leq 
  \lambda_0 (\|\gamma^0_S-\hat{\gamma}_S\|_1 + \|\hat{\gamma}_{S^c}\|_1)  + \frac{4}{3}\lambda \|\gamma^0_S - \hat{\gamma}_S\|_1\\
  \implies\\
    &\|\hat{\gamma}_{S^c}\|_{1} \leq 
  \frac{\lambda_0 + \frac{4}{3}\lambda}{\frac{4}{5}\lambda - \lambda_0}\|\gamma^0_S-\hat{\gamma}_S\|_1.
\end{align*}

If $0 \leq \lambda_0 < \frac{4}{10} \lambda$, it follows $\frac{\lambda_0 + \frac{4}{3}\lambda}{\frac{4}{5}\lambda - \lambda_0} \leq \frac{2\lambda_0 + \frac{4}{3}\lambda}{\frac{4}{5}\lambda - 2\lambda_0}$, such that the choice 
\[\lambda_0 \leq \frac{\lambda}{2(1+C_0)}\left(\frac{4}{5}C_0 - \frac{4}{3}\right)\] 
implies $\hat\gamma - \gamma_0 \in A$, for any $C_0 \geq 5/3$ such that $\lambda_0 \geq 0$. It remains to show, for any $C_0 \geq 5/3$,
\[\frac{1}{2(1+C_0)}\left(\frac{4}{5}C_0 - \frac{4}{3}\right) < \frac{4}{10}.\]
Indeed, rearranging gives 
\[-\left(\frac{8}{10}+\frac{4}{3}\right)  < \left(\frac{8}{10} - \frac{4}{5}\right)C_0 = 0,\]
such that the inequality is satisfied.

\qed

\subsection{Proof of Lemma \ref{lem:nonasymp-prob}}\label{proof:nonasymp-prob}
Under Condition \ref{cond:mixing} it follows
$\{X_{it}X_{jt}\}$ and $\{X_{it}\varepsilon_t\}$ are strictly stationary and mixing with coefficients bounded by $\bar\beta(t)$, for each $i,j\in[p]$.

We first consider event $\mathcal{S}$. Note, by the Mean Value Theorem, for all $i \in [p]$,
$\left|\hat{\sigma}_i - \sigma_i\right| = \left|\sqrt{\hat{\sigma}_i^2} - \sqrt{\sigma_i^2}\right| \leq \left|\tilde{\sigma}_i^{-1}\right| \left|\hat{\sigma}_i^2 - \sigma_i^2
\right| \leq C \left|\hat{\sigma}_i^2 - \sigma_i^2\right|$, for some finite constant $C>0$, where the final inequality follows by $\tilde{\sigma}_i \in (\sigma_i \wedge \hat\sigma_i, \sigma_i \vee \sigma_i)$ such that $\tilde{\sigma}_i\in(0,\infty)$.

Letting $w_t = (X_{1t}^2,\ldots,X_{pt}^2)'$, we obtain from Lemma \ref{lem:fuknagaev}
\begin{align*}
P(\max_{i\in[p]}\left|\hat{\sigma}_i - \sigma_i\right| > \frac{1}{4}) = 
3 p a_T \Biggl(
\frac{C_q^{(1)}}{d_T^{q-1}}
+ 2 \exp\bigl(-C_q^{(2)}d_T\bigr)
\Biggr) + 
2pd_T\bar\beta(a_T).
\end{align*}

Then, consider event $\mathcal{T}$. Letting $w_t = X_t\varepsilon_t$, we obtain from Lemma \ref{lem:fuknagaev}, for some finite constant $C>0$,
\begin{align*}
&P( \max_{i\in[p]} \frac{2}{T}\left|\varepsilon'X_j \right|  > \lambda_0) \\
&\leq P( \max_{i\in[p]} \left|\varepsilon'X_j \right|  > CT \lambda)\\
&\leq 
3 p a_T \Biggl(
\frac{C_q^{(1)}}{\lambda^q d_T^{q-1}}
+  \exp\bigl(-C_q^{(2)}\lambda^2 d_T\bigr)
\Biggr) + 
2pd_T\bar\beta(a_T).
\end{align*}

Finally, we consider event $\mathcal{U}$. Letting $w_t = \text{vec}({X_t X_t'})$, we obtain from Lemma \ref{lem:fuknagaev}, for some finite constant $C>0$,

\begin{align*}
&P\left( \max_{i,j\in [p]} \left| \frac{1}{T}\sum_{t=1}^TX_{it} X_{jt} - E[X_{it} X_{jt}] \right|  > \lambda_1\right) \\
& \leq P\left( \max_{i,j\in [p]} \left| \sum_{t=1}^TX_{it} X_{jt} - E[X_{it} X_{jt}] \right|  > TC/s_0 \right) \\
&\leq 
3 p^2 a_T \Biggl(
\frac{C_q^{(1)}s_0^q}{ d_T^{q-1}}
+  \exp\bigl(-C_q^{(2)} d_T/s_0^2\bigr)
\Biggr) + 
2p^2d_T\bar\beta(a_T).
\end{align*}

The final result follows by the union bound.
\qed

\subsection{Proof of Lemma \ref{lem:augregress-quantile-bound}}\label{proof:augregress-quantile-bound}
First, note 
\begin{align*}
\|\hat{Y} - \tilde{Y}\|_2^2 &= \sum_{t=1}^T |\xi_\tau(Y_t,Q_t) - \xi_\tau(Y_t,\hat{Q}_t)|^2,
\end{align*}
with $\xi_\tau(y,q) := q - \frac{1}{\tau}\mathds{1}(y < q)(y-q)$. 

We may use the inequality $|\xi_\tau(y,q) - \xi_\tau(y,q')| \leq \left(1+\frac{1}{\tau}\right) |q-q'|$, which holds for each $(y,q,q')\in\mathbb{R}^3$. Hence, we obtain the result
\begin{align*}
\|\hat{Y} - \tilde{Y}\|_2^2 &\leq \sum_{t=1}^T \left(1+\frac{1}{\tau}\right)^2\left|Q_t-\hat{Q}_t\right|^2
= \left(1+\frac{1}{\tau}\right)^2 \|\hat{Q} - Q\|_2^2.
\end{align*}

Finally, we verify the inequality $|\xi_\tau(y,q) - \xi_\tau(y,q')| \leq \left(1+\frac{1}{\tau}\right) |q-q'|$, for each $(y,q,q')\in\mathbb{R}^3$ by considering four cases that encompass all orderings of $(y,q,q')$. Case 1: If $y < q$ and $y < q'$, then $|\xi_\tau(y,q) - \xi_\tau(y,q')| = |q-q' -\frac{1}{\tau}(q-q')| \leq \left(1+\frac{1}{\tau}\right) |q-q'|$; Case 2: If $q < y \leq q'$, then $|\xi_\tau(y,q) - \xi_\tau(y,q')| = |q-q' +\frac{1}{\tau}(y-q')| \leq |q-q'| + \frac{1}{\tau}|y-q'| \leq \left(1+\frac{1}{\tau}\right) |q-q'|$, where the first inequality follows by the Triangle Inequality and the second inequality from the ordering $q<y<q'$; Case 3: If $q < y \leq q'$, then $|\xi_\tau(y,q) - \xi_\tau(y,q')| = |q-q' - \frac{1}{\tau}(y-q)| \leq |\left(1+\frac{1}{\tau}\right) |q-q'|$, by similar steps as in Case 2; Case 4: If  $q \leq y$ and $q' \leq y$, then then $|\xi_\tau(y,q) - \xi_\tau(y,q')| = |q-q'| \leq |\left(1+\frac{1}{\tau}\right) |q-q'|$.

\qed

\subsection{Proof of Lemma \ref{lem:quant-consistency}}\label{proof:quant-consistency}

Under the imposed conditions, the Assumptions of Theorem 4.2 in BCMPW are satisfied, which is shown at the end of this proof. 

Then, notice
\begin{align*}
\|\hat{Q} - Q\|_2^2/T 
&= \frac{1}{T}\sum_{t=1}^T \left(X_t'(\hat\alpha - \alpha^0)\right)^2\\
&= (\hat\alpha - \alpha^0)'  \frac{1}{T}\sum_{t=1}^T X_t X_t' (\hat\alpha - \alpha^0)\\
&= (\hat\alpha - \alpha^0)'  \hat\Sigma (\hat\alpha - \alpha^0).
\end{align*}

We may employ the inequality $\alpha'\hat\Sigma\alpha \leq C_1 \alpha'\Sigma \alpha$, for any $\alpha \in A_{\alpha}$, which holds for some universal constant $C_1 > 0$ with probability approaching one, as shown at the end of this proof. It follows that we may bound the RHS (with probability approaching one) by
\begin{align*}
C_1 (\hat\alpha - \alpha^0)'  \Sigma (\hat\alpha - \alpha^0)
= \frac{C_1}{\underline{f}} J_\tau(\hat\alpha - \alpha^0),
\end{align*}
since $\hat\alpha - \alpha^0 \in A_\alpha$ with probability approaching one by Lemma S4 in BCMPW and
where we recall the definition of $J_{\tau}(\cdot)$ in Section \ref{appx:belloni}.

Finally, in the proof to their Theorem 4.2, BCMPW establish 
\[J_\tau(\hat\alpha - \alpha^0)^{1/2} \lesssim
\frac{\psi_T \sqrt{(1+s_{\alpha,0})a_T\log(pa_T)}}{\sqrt{d_T} \kappa_0 \underline{f}^{1/2}},
\]
with probability approaching one, such that 
\begin{align*}
\|\hat{Q} - Q\|_2^2/T = O_P\left(
\frac{\psi_T^2 (1+s_{\alpha,0})a_T\log(pa_T)}{d_T\kappa_0^2 \underline{f}^2}
\right),
\end{align*}
and hence $\|\hat{Q} - Q\|_2^2/T = o_P(1)$ under the rate in Condition \ref{cond:BCMPW}.

We now establish the inequality $\alpha'\hat\Sigma\alpha \leq C_1 \alpha'\Sigma \alpha$, for any $\alpha \in A_{\alpha}$, holds for some universal constant $C_1 > 0$ with probability approaching one. As in the proof of Lemma \ref{lem:nonasymp-bound}, for any $\alpha \in A_\alpha$, under the Condition \ref{cond:BCMPW} and on $\mathcal{U}$,
\begin{align*}
\left|\frac{\alpha'\hat\Sigma\alpha}{\alpha'\Sigma\alpha} - 1\right|
\leq 
\frac{1}{\alpha'\Sigma\alpha} \left|\alpha'(\hat\Sigma - \Sigma)\alpha \right|
\leq 
\frac{1}{\alpha'\Sigma\alpha} \left\|\alpha\right\|_1^2 \|\hat\Sigma - \Sigma\|_\infty
\leq 
\frac{(1+C_0)^2\underline{f}s_0}{\kappa_0^2} \lambda_1 \leq \frac{1}{2}\underline{f} \leq \frac{1}{2}.
\end{align*}
such that $\gamma'\Sigma\gamma \leq 2 \gamma'\hat\Sigma\gamma$ with probability approaching one if $P(\mathcal{U}) \rightarrow 1$ for $\lambda_1 \leq \frac{\kappa_0^2}{2 (1+C_0)^2 \underline{f} s_0}$. From similar steps to those in the proof of Lemma \ref{lem:asymp}, we obtain \[P(\mathcal{U}) \leq 3 p^2 a_T \Biggl(
\frac{C_q^{(1)} \underline{f}^q s_0^q}{ \kappa_0^{2q} d_T^{q-1}}
+  \exp\bigl(-C_q^{(2)}\kappa_0^4 d_T/(s_0^2\underline{f}^2)\bigr)
\Biggr) + 
2p^2d_T\bar\beta(a_T),\]
such that the RHS converges to zero under the rates imposed in Condition \ref{cond:BCMPW}.

Finally, we show the Assumptions 2-4 in BCMPW imposed in their Theorem 4.2 are satisfied. Assumption 2 is satisfied under Conditions \ref{cond:Variances}, \ref{cond:mixing}, \ref{cond:moments}, and \ref{cond:rates}; see Lemma \ref{lem:asymp}. Assumption 3 is satified since we can decompose $Y_t$ as $Y_t = X_t'\alpha^0(\tau) + X_t'\left(\alpha^0(U_t) - \alpha^0(\tau)\right)$. Defining $\epsilon_t = X_t'\left(\alpha^0(U_t) - \alpha^0(\tau)\right)$ it is obvious that $P(\epsilon_t \leq 0 | X_t) = \tau$ from the definition of the vector functional $\alpha^0(\cdot)$. Indeed, by monotonicity of $X_t'\alpha^0(\tau)$ in $\tau$ (a.s.) we obtain,
\begin{align*}
P(\epsilon_t \leq 0 | X_t) &= P\left(X_t'\alpha^0(U_t) \leq X_t'\alpha^0(\tau) | X_t \right) = P\left(U_t \leq \tau \right) = \tau.
\end{align*}
Moreover, given the measurability of $\alpha^0(\cdot)$ and the $U_t$ being iid, it follows $\{\alpha^0(U_t)\}$ is mixing of any rate and therefore $\{Y_t\}$ and $\{\epsilon_t\}$ are mixing of rate $\overline{\beta}(t)$ by Theorem 3.49 in \citet{white2001asymptotic}. Hence, by another application of Theorem 3.49, $\{(X_t,\epsilon_t)\}$ is stationary and $\beta$-mixing of the same rate. Finally, the requirements on the conditional density are reproduced in Condition \ref{cond:BCMPW} in an equivalent form for strictly stationary data.
Assumption 4 is reproduced in Condition \ref{cond:BCMPW} for $m=0$. The bound on $q$ in the statement of Theorem 4.2 is also included.
\qed 

\subsection{Proof of Lemma \ref{lem:fuknagaev}}\label{proof:fuknagaev}

Let $\{\zeta_1,\ldots,\zeta_T\}$ be a sequence with independent blocks, but within blocks identically distributed to the corresponding block of $\{w_1,\ldots,w_T\}$, i.e. \[
\mathcal{L}(\zeta_1,\ldots,\zeta_T) = \mathcal{L}(w_1,\ldots,w_{a_T}) \times \mathcal{L}(w_{a_T+1},\ldots,w_{2a_T}) \times \cdots,\]
where $\mathcal{L}$ denotes the respective law. 


It follows
\begin{align*}
P(\frac{1}{T}|\sum_{t=1}^T w_t| > u) 
&\leq 
P(\frac{1}{T} |\sum_{j=1}^{d_T}\sum_{t\in H_j \cup G_j} w_t | > u/2) + P(\frac{1}{T}|\sum_{t\in Q} w_t| > u/2)\\
&= 
2P(\frac{1}{T} |\sum_{j=1}^{d_T}\sum_{t\in H_j} w_t| > u/4) + P(\frac{1}{T} |\sum_{t\in Q} w_t| > u/2)\\
&\leq 
2P(\frac{1}{T} |\sum_{j=1}^{d_T}\sum_{t\in H_j} \zeta_t| > u/4) + P(\frac{1}{T} |\sum_{t\in Q} \zeta_t| > u/2) + 2d_T\beta(a_T),
\end{align*}
where the first inequality follows by the union bound, and the second inequality by Lemma 4.1 and Lemma 4.2 in \citet{yu1994rates}.

Notice
\begin{align*}
P(\frac{1}{T}|\sum_{j=1}^{d_T}\sum_{t\in H_j} \zeta_t| > u/4)
&\leq 
P(\frac{a_T}{T}\max_{1\leq l \leq a_T}|\sum_{j=1}^{d_T} \zeta_{2(j-1)a_T + l}| > u/4)\\
&\leq 
a_T \max_{1\leq l \leq a_T} P(\frac{a_T}{T} |\sum_{j=1}^{d_T} \zeta_{2(j-1)a_T + l}| > u/4),
\end{align*}
and, similarly, 
\begin{align*}
P(\frac{1}{T}|\sum_{t\in Q} \zeta_t| > u/2)
&\leq 
P(\frac{a_T}{T} \max_{1\leq l \leq a_T}|\zeta_{2d_T a_T + l}| > u/2)\\
&\leq 
a_T \max_{1\leq l \leq a_T} P(\frac{a_T}{T} |\zeta_{2d_T a_T + l}| > u/2).
\end{align*}



Finally, by the Fuk-Nagaev inequality in Corollary 4 in \citet{fuk1971probability}, for $\tilde{u} = \frac{ud_T}{4}$,
\begin{align*}
P(\frac{a_T}{T} |\sum_{j=1}^{d_T} \zeta_{2(j-1)a_T + l}| > u/4)
&\leq P(|\sum_{j=1}^{d_T} \zeta_{2(j-1)a_T + l}| > \tilde{u})\\
&\leq 
\frac{c_q^{(1)}}{\tilde{u}^q} \sum_{j=1}^{d_T} E|\zeta_{2(j-1)a_T + l}|^q
+ \exp(-\frac{c_q^{(2)}\tilde{u}^2}{\sum_{j=1}^{d_T} E|\zeta_{2(j-1)a_T + l}|^2}) \\
&=
\frac{4^q c_q^{(1)}}{u^q d_T^{q-1}} E|\zeta_{2(j-1)a_T + l}|^q
+  \exp(-\frac{c_q^{(2)}u^2d_T}{16 E|\zeta_{2(j-1)a_T + l}|^2}),
\end{align*}
and, similarly, for $\bar{u} = \frac{ud_T}{2}$,
\begin{align*}
P(\frac{a_T}{T} |\zeta_{2d_Ta_T + l}| > u/2)
&\leq P(|\zeta_{2(j-1)a_T + l}| > \bar{u})\\
&\leq 
\frac{c_q^{(1)}}{\bar{u}^q } E|\zeta_{2(j-1)a_T + l}|^q
+  \exp(-\frac{c_q^{(2)}\bar{u}^2}{E|\zeta_{2(j-1)a_T + l}|^2}) \\
&=
\frac{2^q c_q^{(1)}}{u^q d_T^q } E|\zeta_{2(j-1)a_T + l}|^q
+  \exp(-\frac{c_q^{(2)}u^2d_T^2}{4 E|\zeta_{2(j-1)a_T + l}|^2}),
\end{align*}

Combining, and applying the union bound,
\begin{align*}
P(\max_{1\leq i \leq p}|\sum_{t=1}^T w_{t,i}| > u) 
&\leq 
2 p a_T \bigl(
\frac{4^q c_q^{(1)}}{u^q d_T^{q-1}} \max_{1\leq i \leq p}E|w_{1,i}|^q
+  \exp(-\frac{c_q^{(2)}u^2d_T}{16 \max_{1\leq i \leq p}E|w_{1,i}|^2})
\bigr)\\
&\quad +
p a_T  \bigl(
\frac{2^q c_q^{(1)}}{u^qd_T^q} \max_{1\leq i \leq p}E|w_{1,i}|^q
+ \exp(-\frac{c_q^{(2)}u^2 d_T^2}{4 \max_{1\leq i \leq p}E|w_{1,i}|^2})
\bigr)\\
&\quad + 
2pd_T\beta(a_T),
\end{align*}

such that, for some positive constants $C_q^{(1)}$ and  $C_q^{(2)}$, it holds
\begin{align*}
P(\max_{1\leq i \leq p}|\frac{1}{T}\sum_{t=1}^T w_{t,i}| > u) 
&\leq 
3 p a_T \Biggl(
\frac{C_q^{(1)}}{u^q d_T^{q-1}}
+  \exp\biggl(-C_q^{(2)} u^2d_T\biggr)
\Biggr) + 
2pd_T\beta(a_T),
\end{align*}

Hence, there exist uniform constants $C_q^{(3)},C_q^{(4)}>0$, such that, for any $\delta \in (0,1)$,
\begin{align*}
u &= C_q^{(3)} \frac{( pa_T\delta)^{1/q}}{d_T^{(q-1)/q}}\wedge C_q^{(4)} \frac{\sqrt{\log(pa_T/\delta)}}{\sqrt{d_T}}
\end{align*}
it follows 
\begin{align*}
P(\max_{1\leq i \leq p}|\frac{1}{T}\sum_{t=1}^T w_{t,i}| > u) 
&\leq \delta + 2p d_T\beta(a_T).
\end{align*}\qed
\processdelayedfloats





\end{document}